\documentclass[11pt,a4paper]{article}
\usepackage{amsmath,amssymb,multicol,stmaryrd,enumerate,graphicx,color,xspace,fancybox,tikz,paralist,ifthen,verbatim}

\usepackage[margin=3.5cm]{geometry}

\usepackage{listings}
\usepackage{gawlitza}
\usepackage[sort,numbers,sectionbib]{dmnatbib}

\usepackage{amsthm}
\theoremstyle{theorem}
\newtheorem{theorem}{Theorem}
\newtheorem{lemma}[theorem]{Lemma}

\theoremstyle{definition}
\newtheorem{definition}{Definition}

\theoremstyle{remark}
\newtheorem{example}[theorem]{Example}

\newboolean{draft}
\setboolean{draft}{false}
\newcommand{\vspaces}[1]{}

\title{
Improving Strategies via {SMT} Solving
\thanks{\protect\raisebox{-3mm}{\protect\includegraphics[width=6mm]{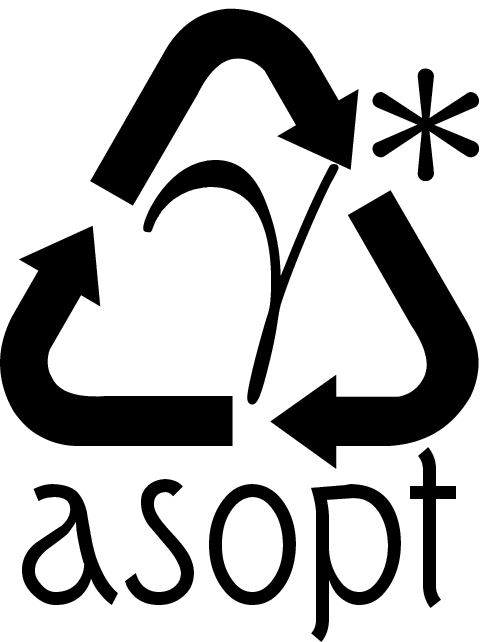}}\quad This work was partially funded by the ANR project ``ASOPT''.}
}

\author{
  Thomas Martin Gawlitza
  \and
  David Monniaux%
  \thanks{Both authors from CNRS (\emph{Centre national de la recherche scientifique}), VERIMAG laboratory. VERIMAG is a joint laboratory of CNRS and Universit\'e Joseph Fourier (Grenoble). Email \emph{firstname}.\emph{lastname}@imag.fr}}

\date{\today}

\usepackage[pdfauthor={Thomas Gawlitza, David Monniaux},
  pdftitle={Improving Strategies via SMT Solving},
  pdfkeywords={static analysis, abstract interpretation, strategy iteration, policy iteration, SMT-solving}]{hyperref}

\renewcommand{\tom}[1]{\ifthenelse{\boolean{draft}}{\nb{#1}}{}}

\newcommand{\np}{\ifthenelse{\boolean{draft}}{\newpage}{}}

\allowdisplaybreaks
\begin{document}
\maketitle
\newcommand{\FIPlayVal}[2]{\PlayVal{#1}{#2}^{\mathrm{filar}}}
\newcommand{\FIGameVal}[2]{\GameVal{#1}{#2}^{\mathrm{filar}}}
\newcommand{\SCPlayVal}[2]{\PlayVal{#1}{#2}^{\mathrm{sc}}}
\newcommand{\SCGameVal}[2]{\GameVal{#1}{#2}^{\mathrm{sc}}}
\newcommand{\CPlayVal}[2]{\PlayVal{#1}{#2}^{\mathrm{c}}}
\newcommand{\CGameVal}[2]{\GameVal{#1}{#2}^{\mathrm{c}}}

\renewcommand{\Prob}{\mathbb{P}}

\renewcommand{\PlayPrefix}{\mathbf{P}^*}
\renewcommand{\Play}{{\mathbf{P}^\omega}}
\newcommand{\Dd}{\mathbb{D}_d}
\newcommand{\CDd}{\overline{\Dd}}
\newcommand{\vw}{\;{\vee\hspace{-2.3mm}\wedge}\;}

\newcommand{\closure}{\mathbf{cl}}

\newcommand{\intersect}{\mathbf{intersect}}

\newcommand{\emptynesstest}{\nabla}

\newcommand{\D}{\mathbb{D}}

\newcommand{\ufinsol}{\mathbf{ufinsol}}

\newcommand{\pos}{\mathrm{pos}}
\newcommand{\Pos}{\mathrm{Pos}}

\newcommand{\stru}[1]{\bigskip \noindent {\bf (* \quad #1 \quad *)} \bigskip}

\newcommand{\CQCP}{\mathbf{CQCP}}
\newcommand{\CQCPABc}{\CQCP_{\A, \mathcal{B}, c}}
\newcommand{\SDP}{\mathbf{SDP}}
\newcommand{\SDPAaBC}{\SDP_{\A, a, \B, C}}
\newcommand{\dom}{\mathrm{dom}}
\newcommand{\fdom}{\mathrm{fdom}}

\renewcommand{\S}[2]{S#1^{#2 \times #2}}

\renewcommand{\SR}[1]{S\R^{#1\times #1}}
\newcommand{\SRp}[1]{S\R^{#1\times #1}_+}
\newcommand{\SRpp}[1]{S\R^{#1\times #1}_{++}}

\newcommand{\A}{\mathcal{A}}
\newcommand{\T}{\mathcal{T}}
\newcommand{\Relaxed}{\mathcal{R}}
\newcommand{\Tr}{\mathrm{Tr}}

\newcommand{\Ext}{\mathrm{Ext}}

\renewcommand{\subsection}[1]{\paragraph{#1}}

\newcommand{\assumption}{{\bf (*)} }

\begin{abstract}
% PERFEKT
We consider the problem of computing numerical invariants of programs by abstract interpretation.
Our method eschews two traditional sources of imprecision:
(i) the use of widening operators for enforcing convergence within
a finite number of iterations (ii) the use of merge operations (often,
convex hulls) at the merge points of the control flow graph. It instead
computes the least inductive invariant expressible in the domain
at a restricted set of program points, and analyzes the rest of the code en bloc. 
We emphasize that we compute this inductive invariant precisely.
For that 
we extend the strategy improvement algorithm of \citet{DBLP:conf/csl/GawlitzaS07}.
If we applied their method directly, 
we would have to solve an exponentially sized system of abstract semantic equations, 
resulting in memory exhaustion.
Instead, we keep the system implicit and discover strategy improvements
using SAT modulo real linear arithmetic (SMT).
For evaluating strategies we use linear programming.
Our algorithm has low polynomial space complexity 
and performs for contrived examples in the worst case exponentially many strategy improvement steps; 
this is unsurprising,  since we show that the associated abstract reachability problem is $\Pi^p_2$-complete.
\end{abstract}

\section{Introduction}
\label{s:intro}

\subsection{Motivation}

% PERFEKT
Static program analysis attempts to derive properties about the
run-time behavior of a program without running the program.
Among interesting properties are the numerical ones:
for instance, that a given variable
$x$ always has a value in the range $[12,41]$ when reaching a given
program point. An analysis solely based on such interval relations at all program points is known as \emph{interval analysis}~\cite{CouCou76}. More refined numerical analyses include, for instance, 
finding for each program point an enclosing polyhedron for the vector of program variables~\cite{DBLP:conf/popl/CousotH78}. 
In addition to obtaining facts about the values of numerical program variables, numerical analyses are used as building blocks for e.g. pointer and shape analyses.

% PERFEKT
However, by Rice's theorem, only trivial properties can be checked automatically \cite{rice}.
In order to check non-trivial properties
we are usually forced to use \emph{abstractions}.
A systematic way for inferring properties automatically 
w.r.t.\ a given abstraction is given through 
the \emph{abstract interpretation} framework of \citet{DBLP:conf/popl/CousotC77}.
This framework \emph{safely over-approximates} the run-time behavior of a program.

% PERFEKT
When using the abstract interpretation framework, we usually have two sources of imprecision.
The first source of imprecision is the abstraction itself: for instance, if the property to be proved needs a non-convex invariant to be established, and our abstraction can only represent convex sets, then we cannot prove the property. Take for instance the 
C-code
\lstinline!y = 0; if (x <= -1 || x >= 1) { if (x == 0) y = 1; }!.
No matter what the values of the variables \lstinline{x} and \lstinline{y} are before the execution of the above C-code,
after the execution the value of \lstinline{y} is $0$.
The invariant $|x| \geq 1$ in the ``then'' branch is not convex, 
and its convex hull includes $x=0$. 
Any static analysis method that computes a convex invariant in this branch will thus 
also include $y = 1$. 
In contrast, our method avoids enforcing convexity, 
except at the heads of loops.

% PERFEKT
The second source of imprecision are the safe but imprecise methods that are used
for solving  the \emph{abstract semantic equations} that describe
the abstract semantics: 
such methods safely over-approximate exact solutions,
but do not return exact solutions in all cases.
The reason is that we are concerned with abstract domains
that contain infinite ascending chains,
in particular if we are interested
in numerical properties:
the complete lattice of all $n$-dimensional closed real intervals,
used for interval analysis, is an example.
The traditional methods 
are based on Kleene fixpoint iteration
which (purely applied) is not guaranteed to terminate in interesting cases.
In order to enforce termination
(for the price of imprecision)
traditional methods make use of  
the widening/narrowing approach of
\citet{DBLP:conf/popl/CousotC77}.
Grossly, widening extrapolates the first iterations
of a sequence to a possible limit, but can easily overshoot the desired result.
In order to avoid this, various tricks are used, including
``widening up to''  \citep[Sec.~3.2]{Halbwachs:CAV93}, 
``delayed'' or with ``thresholds'' \citep{ASTREE:PLDI03}.
However, 
these tricks, although they may help in many practical cases,
are easily thwarted. \citet{DBLP:conf/cav/GopanR06} proposed ``lookahead widening'', which discovers new feasible paths and adapts widening accordingly;
again this method is no panacea.
Furthermore, analyses involving widening are \emph{non-monotonic}:
stronger preconditions can lead to weaker invariants being automatically inferred; 
a rather non-intuitive behaviour. 
Since our method does not use widening at all, it avoids these problems.

\subsection{Our Contribution}

% PERFEKT
We fight both sources of imprecision noted above:

\begin{compactitem}
  \item
    % PERFEKT
    In order to improve the precision of the abstraction,
    we abstract sequences of if-then-else statements without loops
    en bloc. 
    In the above example, 
    we are then able to conclude that $y \neq 0$ holds.
    In other words:
    we abstract sets of states only at the heads of loops, or, more generally,
    at a cut-set of the control-flow graph
    (a cut-set is a set of program points such that
    removing them would cut all loops). 

  \item
    %\tom{PERFEKT}
    Our main technical contribution consists of a practical method for
    precisely computing abstract semantics of affine programs
    w.r.t.\ the template linear constraint domains 
    of \citet{DBLP:conf/vmcai/SankaranarayananSM05},
    with sequences of if-then-else statements which do not contain loops abstracted en bloc.
    Our method is based on a strict generalization of the strategy improvement algorithm of
    \citet{DBLP:conf/esop/GawlitzaS07,DBLP:conf/csl/GawlitzaS07,gawlitza_sas_10}.
    The latter algorithm could be directly applied to the problem we solve in this
    article, but the size of its input would be exponential in the size of
    the program, 
    because we then need to explicitly enumerate all program paths
    between cut-nodes which do not cross other cut-nodes.
    In this article, 
    we give an algorithm with low polynomial memory consumption 
    that uses exponential time in the worst case.
    The basic idea 
    consists in avoiding an explicit 
    enumeration of all paths through sequences of if-then-else-statements which do not contain loops.
    Instead we use a SAT
    modulo real linear arithmetic solver for improving the current
    strategy locally.
    %, i.e., for finding a path that improves the current approximate to the abstract semantics.
    For evaluating each strategy 
    encountered during the strategy iteration,
    we use linear programming.

\item    
% PERFEKT
    As a byproduct of our considerations we show that 
    the corresponding abstract reachability problem
    is $\Pi^p_2$-complete. %,
    In fact, we show that it is $\Pi^p_2$-hard even if the
    loop invariant being computed consists in a single $x \leq C$
    inequality where $x$ is a program variable and $C$ is the parameter of
    the invariant. 
    Hence, exponential worst-case running-time seems to be unavoidable.
\end{compactitem}

\subsection{Related Work}

%\tom{PERFEKT}
Recently,
several alternative approaches for computing numerical invariants
(for instance w.r.t.\ to template linear constraints) were developed:

\subsubsection{Strategy Iteration}

%\tom{PERFEKT}
Strategy iteration 
(also called \emph{policy iteration})
was introduced by Howard for
solving stochastic control problems \cite{Howard,Puterman}
and is also applied to two-players zero-sum games 
\cite{HoffmanKarp,Puri95,Voege00}
or min-max-plus systems \cite{Cochet99}.
\citet{Costan05,DBLP:conf/esop/GaubertGTZ07,DBLP:conf/esop/AdjeGG10}
developed a strategy iteration approach for solving the 
abstract semantic equations that occur in static program analysis
by abstract interpretation.
Their approach can be seen as an alternative to 
the traditional widening/narrowing approach.
The goal of their algorithm is to compute least fixpoints of monotone self-maps $f$,
where $f(x) = \min \; \{ \pi(x) \mid \pi \in \Pi \}$ for all $x$ and $\Pi$ is a family of self-maps.
The assumption is that one can efficiently compute the least fixpoint $\mu \pi$ of $\pi$ for every $\pi \in \Pi$.
The $\pi$'s are the (min-)strategies.
Starting with an arbitrary min-stratgy $\pi^{(0)}$,
the min-strategy is successively improved.
The sequence $(\pi^{(k)})_k$ of attained min-strategies results in a decreasing sequence
$\mu \pi^{(0)} > \mu \pi^{(1)} > \cdots > \mu \pi^{(k)}$
% Here, $\mu g$ denotes the least fixpoint of a monotone self-map $g$.
that stabilizes,
% The algorithm stops, 
whenever $\mu \pi^{(k)}$ is a fixpoint of $f$
---
not necessarily the least one.
However, 
there are indeed important cases,
where minimality of the obtained fixpoint can be guaranteed
\cite{2008arXiv0806.1160A}.
Moreover,
an important advantage of their algorithm is that it can be stopped at any time with a safe over-approximation.
This is in particular interesting if there are infinitely many min-strategies \cite{DBLP:conf/esop/AdjeGG10}.
\citet{Costan05} showed how to use their framework for 
performing interval analysis without widening.
\citet{DBLP:conf/esop/GaubertGTZ07} 
extended this work to the following \emph{relational} abstract domains:
The \emph{zone domain} \cite{DBLP:conf/pado/Mine01}, 
the \emph{octagon domain} \cite{DBLP:conf/wcre/Mine01} and 
in particular the
\emph{template linear constraint domains} 
\cite{DBLP:conf/vmcai/SankaranarayananSM05}.
\citet{DBLP:conf/csl/GawlitzaS07} presented 
a practical (max-)strategy improvement algorithm 
for computing least solutions of  
\emph{systems of rational equations}.
Their algorithm enables them to perform 
a template linear constraint analysis \emph{precisely}
---
even if the mappings are not non-expansive.
This means: Their algorithm always computes \emph{least} solutions of 
abstract semantic equations --- not just some solutions.

\subsubsection{Acceleration Techniques}

%\tom{PERFEKT}
\citet{DBLP:conf/sas/GonnordH06,Gonnord_PhD}
investigated an improvement of linear relation analysis that consists in computing,
when possible,
the exact (abstract) effect of a loop.
%\tom{TODO: Is this correct?}
The technique is fully compatible with the use of widening, 
and whenever it applies, 
it improves both the precision and the performance of the analysis.
\citet{LEROUX-SUTRE-SAS2007,DBLP:conf/birthday/GawlitzaLRSSW09}
studied cases where interval analysis can be done in polynomial time w.r.t.\ 
a uniform cost measure,
where memory accesses and arithmetic operations are counted for $\mathcal{O}(1)$.
%
%\citet{Gonnord_Monniaux_POPL11_submission} have proposed \emph{path focusing} using SMT-solving as a way to improve precision of widening-based analyses and to increase opportunities for acceleration. Their method is in a sense more general than the one in the present article, since it applies to a wider class of domains, but it does not guarantee minimality.

\subsubsection{Quantifier Elimination}

% PERFEKT
Recent improvements in SAT/SMT solving techniques have made it possible to perform quantifier elimination on larger formulas \citep{DBLP:conf/lpar/Monniaux08}.
\citet{DBLP:conf/popl/Monniaux09}
developed an analysis method based on quantifier elimination in the theory of
rational linear arithmetic. This method targets the same domains as the present article; it however produces a richer result. 
It can not only compute the least invariant inside the abstract domain of a loop, but also express it as a function of the precondition of the loop; the method outputs the source code of the optimal abstract transformer mapping the precondition to the invariant. Its drawback is its high cost, which makes it practical only on small code fragments; thus, its intended application is \emph{modular analysis}: analyze very precisely small portions of code (functions, modules, nodes of a reactive data-flow program,~\ldots), and use the results for analyzing larger portions, perhaps with another method, 
including the method proposed in this article.

\subsubsection{Mathematical Programming}

%\tom{PERFEKT}
\citet{Colon_CAV03,Sankaranarayanan_SAS04,Cousot05-VMCAI} presented approaches for generating linear invariants 
that uses \emph{non-linear constraint solving}.
\citet{leo09} propose a mathematical programming formulation 
whose constraints define the space of all post-solutions of the abstract semantic equations.
The objective function aims at minimizing the result.
For programs that use affine assignments and affine guards, only,
this yields a \emph{mixed integer linear programming} formulation for interval analysis.
The resulting mathematical programming problems can then be solved to guaranteed global optimality by means of general purpose branch-and-bound type algorithms.

\section{Basics}
\label{s:basics}

\subsection{Notations}

%\tom{PERFEKT}
$\mathbb{B} = \{ 0, 1 \}$ denotes the set of Boolean values.
The set of real numbers %(resp.\ the set of rational numbers) 
is denoted by $\R$. % (resp.\ $\Q$).
The complete linearly ordered set $\R \cup \{ \neginfty, \infty \}$
is denoted by $\CR$.
% Additionally, we set $\CQ := \Q \cup \{ \neginfty, \infty \}$.
We call two vectors $x, y \in \CR^n$ \emph{comparable}
iff $x \leq y$ or $y \leq x$ holds.
For $f : X \to \CR^m$ with $X \subseteq \CR^n$, we set
$\dom(f) := \{ x \in X \mid f(x) \in \R^m \}$
and
$\fdom(f) := \dom(f) \cap \R^n$.
We denote the $i$-th row (resp.\ the $j$-th column) 
of a matrix $A$ by $A_{i\cdot}$ (resp.\ $A_{\cdot j}$).
Accordingly, 
$A_{i \cdot j}$ denotes the component in the
$i$-th row and the $j$-th column.
We also use this notation for vectors and mappings $f : X \to Y^k$.
 
%\tom{PERFEKT}
Assume that a fixed set $\vX$ of variables 
and a domain $\D$ is given.
We consider equations of the form 
$\vx = e$,
where $\vx \in \vX$ is a variable
and $e$ is an expression over $\D$.
A \emph{system} $\E$ of (fixpoint) equations is a finite
set $\{ \vx_1 = e_1,\ldots,\vx_n = e_n \}$ 
of equations, 
where
$\vx_1,\ldots,\vx_n$ are pairwise distinct variables.
We denote the set $\{\varx_1,\ldots,\varx_n\}$ of variables 
occurring in $\E$ by $\vX_\E$.
We drop the subscript whenever it is clear from the context.

%\tom{PERFEKT}
For a variable assignment $\rho : \vX \to \D$,
an expression $e$ is mapped to a value 
$\sem{e}\rho$ 
by setting 
$	  \sem{\vx}\rho := \rho(\vx) $
 and 
$  \sem{f(e_1,\ldots,e_k)}\rho := f(\sem{e_1}\rho,\ldots,\sem{e_k}\rho) $,
where $\vx \in \vX$, $f$ is a $k$-ary operator, for instance $+$, and
$e_1,\ldots,e_k$ are expressions.
Let $\E$ be a system of equations.
We define the unary operator 
$\sem{\E}$ 
on 
$\vX \to \D$
by setting
$ %\begin{align*}
  (\sem{\E}\rho)(\vx) := \sem{e}\rho
$
for all 
$
  \vx = e \in \E
$.
 %\end{align*}
A solution is a variable assignment $\rho$
such that $\rho = \sem{\E}\rho$ 
holds.
The set of solutions  
is denoted 
by $\Sol(\E)$.
 
%\tom{PERFEKT}
Let $\D$ be a complete lattice.
We denote the \emph{least upper bound} and 
the \emph{greatest lower bound} of 
a set $X \subseteq \D$ by $\bigvee X$ and $\bigwedge X$, 
respectively.
The least element $\bigvee \emptyset$
(resp.\ the greatest element $\bigwedge \emptyset$)
is denoted by $\bot$ (resp.\ $\top$).
%Accordingly,
We define the binary operators $\vee$ and $\wedge$
by 
$x \vee y := \bigvee \{ x, y \}$
and
$x \wedge y := \bigwedge \{ x, y \}$ for all $x,y \in \D$,
respectively.
For $\Box \in \{ \vee, \wedge \}$,
we will also consider $x_1 \;\Box\; \cdots \;\Box\; x_k$
as the application of a $k$-ary operator.
This will cause no problems,
since the binary operators $\vee$ and $\wedge$ 
are associative and commutative.
An expression $e$ (resp.\ an equation $\vx = e$) is called 
\emph{monotone}
iff 
all operators occurring in $e$ are monotone.

%\tom{PERFEKT}
The set $\vX \to \D$ of all \emph{variable assignments} 
is a complete lattice.
For $\rho, \rho' : \vX \to \D$,
we write $\rho \ll \rho'$ (resp.\ $\rho \gg \rho'$) 
iff 
$\rho(\vx) < \rho'(\vx)$ (resp.\ $\rho(\vx) > \rho'(\vx)$) holds for all $\vx \in \vX$.
For $d \in \D$,
$\underline d$ denotes the variable assignment 
$\{ \vx \mapsto d \mid \vx \in \vX \}$.
A variable assignment $\rho$ with $\botvar \ll \rho \ll \topvar$
is called \emph{finite}.
A pre-solution (resp.\ post-solution) is a variable assignment $\rho$
such that  
$\rho \leq \sem{\E}\rho$ (resp.\ $\rho \geq \sem{\E}\rho$) holds.
The set of all pre-solutions 
(resp.\ the set of all post-solutions) 
is denoted 
by $\PreSol(\E)$ (resp.\ $\PostSol(\E)$).
The least fixpoint (resp.\ the greatest fixpoint)
of an operator 
$f : \D \to \D$ 
is denoted by $\mu f$ (resp.\ $\nu f$),
provided that it exists.
Thus, the least solution (resp.\ the greatest solution)
of a system $\E$ of equations is denoted by
$\mu\sem\E$ (resp.\ $\nu\sem\E$),
provided that it exists.
For a pre-solution $\rho$ (resp.\ for a post-solution $\rho$),
$\mu_{\geq \rho}\sem\E$ (resp.\ $\nu_{\leq \rho}\sem\E$)
denotes the least solution that is greater than or equal to $\rho$
(resp.\ the greatest solution that is less than or equal to $\rho$).
From Knaster-Tarski's fixpoint theorem we get:
  Every system $\E$ of monotone equations over a complete lattice 
  has a least solution $\mu\sem\E$ and a greatest solution $\nu\sem\E$.
  Furthermore,
  $\mu\sem\E = \bigwedge \PostSol(\E)$
  and
  $\nu\sem\E = \bigvee \PreSol(\E)$.

\subsection{Linear Programming}

% PERFEKT
We consider linear programming problems (LP problems for short)
of the form 
$
  \sup \;
  \{ 
  c^\top  x \mid 
  x \in \R^n
  ,
  Ax \leq b
  \}
  ,
$
  where 
  $A \in \R^{m \times n}$,
  $b \in \R^m$, and
  $c \in \R^n$
  are the inputs.
  The convex closed polyhedron
 $
  \{ 
  x \in \R^n
  \mid
  Ax \leq b
  \}
  $
  is called the \emph{feasible space}.
  The LP problem is called \emph{infeasible}
  iff the feasible space is empty.
  An element of the feasible space,
  is called \emph{feasible solution}.
  A feasible solution $x$ that maximizes $c^\top x$ is called
  \emph{optimal solution}.

% PERFEKT
LP problems can be solved in polynomial time
through interior point methods \cite{Megiddo87,LP1}. 
Note, however, that the running-time then
crucially depends on the sizes of occurring numbers.
At the danger of an exponential running-time in contrived cases, we can also instead rely on
the simplex algorithm: its running-time is \emph{uniform},
i.e., independent of the sizes of occurring numbers
(given that arithmetic operations, comparison,
storage and retrieval for numbers are counted for ${\cal O}(1)$).

\subsection{SAT modulo real linear arithmetic}

%\tom{PEFEKT}
The set of SAT modulo real linear arithmetic formulas $\Phi$ 
is defined through the % following
grammar
%
%\begin{align*}
%  e      &::= c \mid x \mid e_1 + e_2 \mid c \cdot e'  &
%  \Phi &::= a \mid e_1 \leq e_2 \mid \Phi_1 \vee \Phi_2 \mid \Phi_1 \wedge \Phi_2 \mid \overline {\Phi'} 
%\end{align*}
  $e      ::= c \mid x \mid e_1 + e_2 \mid c \cdot e'$, 
  \;
  $\Phi ::= a \mid e_1 \leq e_2 \mid \Phi_1 \vee \Phi_2 \mid \Phi_1 \wedge \Phi_2 \mid \overline {\Phi'}$.
Here, $c \in \R$ is a constant, $x$ is a real valued variable, $e, e',e_1,e_2$ are real-valued linear expressions, 
$a$ is a Boolean variable and $\Phi, \Phi',\Phi_1,\Phi_2$ are formulas.
An \emph{interpretation} $I$ for a formula $\Phi$ is a mapping that assigns a real value to every real-valued variable and
a Boolean value to every Boolean variable.
We write $I \models \Phi$ for ``$I$ is a \emph{model} of $\Phi$'',
i.e.,
%\begin{align*}
%    \sem{c} I &= c 
%    & 
%    \sem{x} &= I(x)
%    &
%    \sem{e_1 + e_2}I &= \sem{e_1}I + \sem{e_2}I
%    &
%    \sem{c \cdot e'}I &= c \cdot \sem{e'}I
%\end{align*}
    $\sem{c} I = c$,
    $\sem{x} = I(x)$,
    $\sem{e_1 + e_2}I = \sem{e_1}I + \sem{e_2}I$, 
    $\sem{c \cdot e'}I = c \cdot \sem{e'}I$, and:
\begin{align*}
    I \models a &\iff I(a) = 1
    &
    I \models e_1 \leq e_2 &\iff \sem{e_1}I \leq \sem{e_2}I
    \\[-1mm]
    I \models \Phi_1 \vee \Phi_2 &\iff I \models \Phi_1 \text{ or } I \models \Phi_2
    &
    I \models \Phi_1 \wedge \Phi_2 &\iff I \models \Phi_1 \text{ and } I \models \Phi_2
    \\[-1mm]
    I \models \overline{\Phi'} &\iff I \not\models \Phi'
\end{align*}

\noindent
A formula is called \emph{satisfiable} iff it has a model.
The problem of deciding, whether or not a given SAT modulo real linear arithmetic formula is 
satisfiable, is NP-complete.
There nevertheless exist efficient solver implementations for this decision problem \cite{DBLP:conf/cav/DutertreM06}.

%\tom{PERFEKT}
In order to simplify notations we also allow matrices, vectors, the operations $\geq, \break <, >, \neq, =$, 
and the Boolean constants $0$ and $1$ to occur. % in SAT modulo real linear arithmetic formulas.

\subsection{Collecting and Abstract Semantics}

% PERFEKT
The programs that we consider in this article
use real-valued variables $x_1,\ldots,x_n$.
Accordingly, we denote by $x=(x_1,\ldots,x_n)^\top$ the vector of all program variables.
For simplicity, 
we only consider elementary statements of the form 
$x := Ax + b$, and $A x \leq b$,
where $A \in \R^{n \times n}$ (resp.\ $\R^{k \times n}$),
$b \in \R^n$ (resp.\ $\R^{k}$), and
$x  \in \R^n$ denotes the vector of all program variables.
Statements of the form $x := Ax + b$ are called \emph{(affine) assignments}.
Statements of the form $A x \leq b$ are called \emph{(affine) guards}.
Additionally,
we allow statements of the form 
$s_1;\cdots;s_k$ and $s_1 \mid \cdots \mid s_k$,
where $s_1,\ldots,s_k$ are statements.
The operator $;$ binds tighter than the
operator $\mid$, and we consider $;$ and $\mid$ to be right-associative,
i.e., $s_1 \mid s_2 \mid s_3$ stands for $s_1 \mid (s_2 \mid s_3)$,
and $s_1 ; s_2 ; s_3$ stands for $s_1 ; (s_2 ; s_3)$. 
The set of statements is denoted by $\Stmt$.
A statement of the form $s_1 \mid \cdots \mid s_k$,
where $s_i$ does not contain the operator $\mid$ for all $i = 1,\ldots,k$, is called 
\emph{merge-simple}.
A merge-simple statement $s$ that does not use the $\mid$ operator at all 
is called \emph{sequential}.
A statement is called \emph{elementary} iff
it neither contains the operator $\mid$ nor the operator $;$.

% PERFEKT
The \emph{collecting semantics} $\sem{s} : 2^{\R^n} \to 2^{\R^n}$ of 
a statement $s \in \Stmt$ is defined by
\begin{align*}
  \sem{x := Ax + b}X 
    &:= \{ Ax + b \mid x \in X \},
  &
  \sem{A x \leq b}X 
    &:= \{ x \in X \mid A x \leq b \},
  \\[-0.5mm]
  \sem{s_1;\cdots; s_k}  
    &:= \sem{s_k} \circ \cdots \circ \sem{s_1}
  &
  \sem{s_1 \mid\cdots\mid s_k} X
    & := \sem{s_1}X \cup \cdots \cup \sem{s_k}X
\end{align*}
for $X \subseteq \R^n$.
Note that the operators $;$ and $|$ are associative,
i.e., $\sem{(s_1;s_2);s_3} = \sem{s_1;(s_2;s_3)}$
and $\sem{(s_1 \mid s_2) \mid s_3} = \sem{s_1 \mid (s_2 \mid s_3)}$ 
hold for all statements $s_1,s_2,s_3$.

% PERFEKT
An \emph{(affine) program} $G$ is a triple $(N,E,\start)$,
where $N$ is a finite set of \emph{program points},
$E \subseteq N \times \Stmt \times N$ is a finite set of control-flow edges,
and $\start \in N$ is the \emph{start program point}.
As usual, 
the \emph{collecting semantics} $\Values$ of
a program $G = (N,E,\start)$ is the least solution of 
the following constraint system:
\begin{align*}
  \VALUES[\start]
    &\supseteq \R^n
  \qquad
  \VALUES[v] 
    \supseteq \sem{s} (\VALUES[u])
  \quad \text{for all } (u,s,v) \in E
\end{align*}
Here, the variables $\VALUES[v]$, $v \in N$ 
take values in $2^{\R^n}$.
The components of 
the collecting semantics $\Values$ 
are denoted by $\Values[v]$
for all
$v \in N$.

% PERFEKT
Let $\D$ be a complete lattice 
(for instance the complete lattice of all $n$-dimensional closed real intervals).
Let the partial order of $\D$ be denoted by $\leq$.
Assume that $\alpha : 2^{\R^n} \to \D$ 
and 
$\gamma : \D \to 2^{\R^n}$
form a Galois connection,
i.e.,
for all $X \subseteq \R^n$ and all $d \in \D$,
$\alpha(X) \leq d$ iff $X \subseteq \gamma(d)$.
The \emph{abstract semantics} $\sem{s}^\sharp : \D \to \D$ 
of a statement $s$
is defined by 
$ 
  \sem{s}^\sharp := \alpha \circ \sem{s} \circ \gamma
  .
$ 
The \emph{abstract semantics} $\Values^\sharp$ of an affine program
$G = (N,E,\start)$ is the least solution of the following 
constraint system:
\begin{align*}
  \VALUES^\sharp[\start]
    &\geq \alpha(\R^n)
  \qquad
  \VALUES^\sharp[v] 
    \geq \sem{s}^\sharp (\VALUES^\sharp[u])
  \quad \text{for all } (u,s,v) \in E
\end{align*}
Here, the variables $\VALUES^\sharp[v]$, $v \in N$ 
take values in $\D$.
The components of 
the abstract semantics $\Values^\sharp$ 
are denoted by $\Values^\sharp[v]$
for all $v \in N$.
The abstract semantics $\Values^\sharp$ safely over-approximates the collecting semantics $\Values$,
i.e., $\gamma(\Values^\sharp[v]) \supseteq \Values[v]$ for all $v \in N$.

\subsection{Using Cut-Sets to improve Precision}

%\tom{PERFEKT, TODO: Ueberpruefen}
Usually, only sequential statements (these statements correspond to \emph{basic blocks}) 
are allowed in control flow graphs.
However, given a cut-set $C$, one can systematically transform any control flow graph $G$
into an equivalent control flow graph $G'$ of our form 
(up to the fact that $G'$ has fewer program points than $G$)
with increased precision of the abstract semantics.
%This can be done as follows:
%For any two cut-points $n_1,n_2 \in C$ (not necessarily different),
%we consider the sub-graph of $G$ that is induced by all paths 
%from $n_1$ to $n_2$ that do not go through other cut-points.
%Such a sub-graph represents a statement. 
%This statement can be computed in linear time (we assume sharing within expressions).
%We thus can replace every such sub-graph by the respective statement
%in order to obtain the affine program $G'$ of our form 
%that is equivalent 
%in the collecting semantics,
%but is at least as precise and in many cases more precise in the abstract semantics.
%Given a control-flow graph $G$ and a cut-set $C$ of $G$ 
%such a rewriting can be performed in cubic time.
%We could avoid this pre-processing phase and work directly on the original control-flow $G$ graph using the cut-set $C$.
%However, for the sake of simplicity, we don't do so in the present paper.
However, for the sake of simplicity, we do not discuss these aspects in detail.
Instead, we consider an example:

% DM REMARK
% This transformation only seems possible in linear time if the control-flow graph between the cut-nodes is well-structured. Broken by constructs such as:
%  /\
% /  \
% \  /\
%  \/  \
%   \  /
%    \/
\tikzstyle{point}=[circle,draw,thick,inner sep=1pt,minimum size=2mm]

\begin{figure}
  \vspaces{-5mm}
  \centering
	\begin{tabular}{c@{\;}c}
	\scalebox{0.9}{
	\begin{tikzpicture}
		 \node (start) [point] {$\start$};
		 \node (n1) [below of = start,point,yshift=1mm]{$1$};
		 \node (n2) [below of = n1,point,yshift=3mm]{$2$};
		 \node (n3) [below of = n2,point,yshift=3mm]{$3$};
		 \node (n4) [left of = n2,point,xshift=-15mm]{$4$};
		 \node (n4a) [coordinate,left of = n3,xshift=-15mm]{};
		 \node (n4b) [coordinate,left of = n1,xshift=-15mm]{};
		 \node (n5) [right of = n2,point,xshift= 15mm]{$5$};
		 \node (n5a) [coordinate,right of = n3,xshift= 15mm]{};
		 \node (n5b) [coordinate,right of = n1,xshift= 15mm]{};
		 \path[->] (start) edge [] node [right,yshift=1mm] {$x_1 := 0$} (n1);
		 \path[->] (n1) edge [] node [right] {$x_1 \leq 1000$} (n2);
		 \path[->] (n2) edge [] node [right] {$x_2 := - x_1$} (n3);
		 \path[-] (n3) edge [] node [below] {$x_2 \leq -1$} (n4a);
		 \path[->] (n4a) edge [] node [left] {} (n4);
		 \path[-] (n4) edge [] node [left] {} (n4b);
		 \path[->] (n4b) edge [] node [above] {$x_1 := -2 x_1$} (n1);
		 \path[-] (n3) edge [] node [below] {$x_2 \geq 0$} (n5a);
		 \path[->] (n5a) edge [] node [left] {} (n5);
		 \path[-] (n5) edge [] node [left] {} (n5b);
		 \path[->] (n5b) edge [] node [above] {$x_1 := -x_1 + 1$} (n1);
	\end{tikzpicture}
	}
	&
	\scalebox{0.9}{
	\begin{tikzpicture}
		 \node (start) [point] {$\start$};
		 \node (n1) [below of = start,point,yshift=1mm]{$1$};
		 \node (n2) [coordinate,below of = n1,yshift=0mm]{$2$};
		 \node (n3) [coordinate,left of = n2]{$2$};
		 \node (n4) [coordinate,left of = n1]{$2$};
		 \path[->] (start) edge [] node [right] {$x_1 := 0$} (n1);
		 \path[-] (n1) edge [] node [right] 
		   {$
		     \begin{array}{l}
		       x_1 \leq 1000; x_2 := -x_1; \\
		       (x_2 \leq -1; x_1 := -2x_1 \mid x_2 \geq 0; x_1 := -x_1 + 1)
		     \end{array}
		     $} 
		   (n2);
		 \path[-] (n2) edge [] node [right] {} (n3);
		 \path[-] (n3) edge [] node [right] {} (n4);
		 \path[->] (n4) edge [] node [right] {} (n1);
	\end{tikzpicture}
	}
	\\[-2mm]
	(a) & (b)
	\end{tabular}
	\vspaces{-3mm}
	\caption{}
	\label{fig:run:ex:01}
	\vspaces{-8mm}
\end{figure}

\begin{example}[Using Cut-Sets to improve Precision]
  \label{ex:running:0}
  %\tom{TODO}
  As a running example throughout the present article 
  we use the following C-code:
\begin{lstlisting}
int x_1, x_2; x_1 = 0; while (x_1 <= 1000) { x_2 = -x_1; 
  if (x_2 < 0) x_1 = -2 * x_1; else x_1 = -x_1 + 1; }
\end{lstlisting}
This C-code is abstracted through the affine program $G_1 = (N_1,E_1,\start)$
which is shown in Figure \ref{fig:run:ex:01}.(a).
However, it is unnecessary to apply abstraction at every program point;
it suffices to apply abstraction at a cut-set of $G_1$.
Since all loops contain program point $1$,
a cut-set of $G_1$ is $\{ 1 \}$.
Equivalent to applying abstraction only at program point $1$ is
to rewrite the control-flow graph w.r.t.\ the cut-set $\{ 1 \}$
into a control-flow graph $G$ equivalent w.r.t.\ the collecting semantic.
The result of this transformation is drawn in Figure \ref{fig:run:ex:01}.(b).
  This means:
  the affine program for the above C-code is
  $G = (N,E,\start)$, where
  $ %\begin{align*}
    N = \{ \start, 1 \}
    ,% \quad
    E = \{ (\start, x_1 := 0, 1), (1, s, 1) \}
    ,
  $ %\end{align*}
  and
  \begin{align*}
    s'     &=  x_1 \leq 1000 ; x_2 := -x_1&
    s_1 &= x_2 \leq -1; x_1 := -2 x_1 \\
    s_2 &= -x_2 \leq 0; x_1 := -x_1+1 &
    s      &= s';  (s_1 \mid s_2 )
  \end{align*}
  Let $\Values_1$ denote the collecting semantics of $G_1$
  and $\Values$ denote the collecting semantics of $G$.
  $G_1$ and $G$ are equivalent in the following sense:
  $\Values[v] = \Values_1[v]$ holds for all program points $v \in N$.
  W.r.t.\ the abstract semantics, $G$ is, is we will see, strictly more precise than $G_1$.
  In general we at least have $\Values^\sharp[v] \subseteq \Values^\sharp_{1\cdot}[v]$ for all program points $v \in N$.
%  where $\Values^\sharp_{1\cdot}$ denotes the abstract semantics of $G_1$
%  and $\Values^\sharp$ denotes the abstract semantics of $G$.
  This is independent of the abstract domain.\footnote{We assume that we have given a Galois-connection and thus in particular monotone best abstract transformers.}
\qed
\end{example}

\subsection{Template Linear Constraints}

%\tom{PERFEKT}
In the present article we restrict our considerations to \emph{template linear constraint domains}~\citep{DBLP:conf/vmcai/SankaranarayananSM05}.
Assume that we are given a fixed 
\emph{template constraint matrix} $T \in \R^{m \times n}$.
The template linear constraint domain is  
$\CR^m$.
As shown by \citet{DBLP:conf/vmcai/SankaranarayananSM05},
the concretization $\gamma : \CR^m \to 2^{\R^n}$
and
the abstraction $\alpha : 2^{\R^n} \to \CR^m$,
which are defined by
\begin{align*}
  \gamma(d) &:= \{ x \in \R^n \mid Tx \leq d \} 
    && \forall d \in \CR^m
  ,
  \\
  \alpha(X) &:= \textstyle\bigwedge \{ d \in \CR^m \mid \gamma(d) \supseteq X \}
  && \forall X \subseteq \R^n
    ,
\end{align*}
form a Galois connection.
The template linear constraint domains contain \emph{intervals},
\emph{zones}, and \emph{octagons}, with appropriate choices of the
template constraint matrix \cite{DBLP:conf/vmcai/SankaranarayananSM05}.

%%\tom{OK}
%We now have to discuss how to compute 
%$\sem{s}^\sharp d$ 
%for a given statement $s$ and a given abstract value $d \in \R^m$.
%
%\tom{PERFEKT}
In a first stage we restrict our considerations to sequential 
and merge-simple statements.
Even for these statements we 
avoid unnecessary imprecision,
if we abstract such statements en bloc instead of 
abstracting each elementary statement separately:

\begin{example}
  %\tom{OK}
  In this example we use the interval domain as abstract domain,
  i.e., our complete lattice consists of all $n$-dimensional closed real intervals.
  Our affine program will use $2$ variables, i.e., $n = 2$.
  The complete lattice of all $2$-dimensional closed real intervals can be specified 
  through the template constraint matrix
  $T = \begin{pmatrix} -I &\; I \end{pmatrix}^\top  \in \R^{4 \times 2}$,
  where $I$ denotes the identity matrix.
  %
%  \begin{align*}
%    T 
%    =
%    \begin{pmatrix}
%      -1 & 0 \\
%      1 & 0 \\
%      0 & -1 \\
%      0 & 1 
%    \end{pmatrix}
%  \end{align*}
%  Let us consider the following C-code:
%\lstinline{x_2 = x_1; x_1 = x_1 - x_2;}
%After executing this piece of code,
%the variable \lstinline{x_1} has the value $0$.
%  Restated in our formalism we are 
  Consider the statements 
%  We consider the statements
  %
    $s_1 = x_2 := x_1$, 
    $s_2 = x_1 := x_1 - x_2$, and
    $     s = s_1 ; s_2$
  and the abstract value $I = [0,1] \times \R$ (a $2$-dimensional closed real interval).
  The interval $I$ can w.r.t.\ $T$ be identified with the abstract value 
  $(0,\infty,1,\infty)^\top$.
  More generally,
  w.r.t.\ $T$ every $2$-dimensional closed real interval $[l_1,u_1] \times [l_2,u_2]$
  can be identified with the abstract value $(-l_1,-l_2,u_1,u_2)^\top$.
  If we abstract each elementary statement separately,
  then we in fact use
  $\sem{s_2}^\sharp \circ \sem{s_1}^\sharp$
  instead of $\sem{s}^\sharp$
  to abstract the collecting semantics $\sem s$ of the statement $s = s_1; s_2$.
  The following calculation shows that this can be important:
$ %  \begin{align*}
    \sem{s}^\sharp I
    =
    [0,0] \times [0,1]
    \neq
    [-1,1] \times [0,1]
    =
    \sem{s_2}^\sharp ([0,1] \times [0,1])
    =
    (\sem{s_2}^\sharp \circ \sem{s_1}^\sharp) I
    .
$ %  \end{align*}
  The imprecision is caused by the additional abstraction.
  We lose 
  the information that the values of the program variables 
  $x_1$ and $x_2$ are equal after executing the first statement.
  \qed
\end{example}

% PERFEKT
\noindent
Another possibility for avoiding unnecessary imprecision in the above example would 
consist in adding additional rows to the template constraint matrix.
Although this works for the above example, it does not work in general,
since still only convex sets can be described,
but sometimes non-convex sets are required
(cf.\ with the example in the introduction).

%\tom{PERFEKT}
Provided that $s$ is a merge-simple statement,
$\sem{s}^\sharp d$ 
can be computed in polynomial time through linear programming:

\begin{lemma}[Merge-Simple Statements]
\label{l:merge-simple:poly}
  %\tom{PERFEKT}
  Let $s$ be a merge-simple statement and $d \in \CR^m$.
  Then $\sem{s}^\sharp d$ can be computed in polynomial time
  through linear programming.
  \qed
\end{lemma}

%\tom{PERFEKT}
\noindent
However,
the situation for arbitrary statements is significantly more difficult,
since,
by reducing SAT to the corresponding decision problem,
we can show the following:

\begin{lemma}
\label{l:diamand:is:np:complete}
%\tom{PERFEKT}
  The problem of deciding,
  whether or not,
  for a given template constraint matrix $T$,
  and a given statement $s$,
  $
    \sem{s}^\sharp \inftyvar > \neginftyvar
  $
  holds,
  is NP-complete.
\end{lemma}

%\tom{PERFEKT}
\noindent
Before proving the above lemma,
we introduce $\vee$-strategies for statements as follows:

\begin{definition}[$\vee$-Strategies for Statements]
  %\tom{PERFEKT}
  A $\vee$-strategy $\sigma$ for a statement $s$ is a function
  that maps every position of a $\mid$-statement, 
  (a statement of the form $s_0 \mid s_1$)
  within $s$
  to $0$ or $1$.
  The application $s \sigma$ of a $\vee$-strategy $\sigma$ to a statement $s$ is
  inductively defined by
$ %  \begin{align*}
    s \sigma = s$,
    $(s_0 \mid s_1) \sigma = s_{\sigma(\pos(s_0 \mid s_1))} \sigma$, and
    $(s_0 ; s_1) \sigma = (s_0 \sigma ; s_1 \sigma)
$, %  \end{align*}
  where $s$ is an elementary statement, and $s_0,s_1$ are arbitrary statements.
  %Here, 
  For all occurrences $s'$, $\pos(s')$ denotes the position of $s'$,
  i.e., $\pos(s')$ identifies the occurrence.
  \qed
\end{definition}

\begin{proof}
%\tom{PERFEKT}
  Firstly,
  we show containment in $\mathrm{NP}$.
  Assume 
  $\sem{s}^\sharp \inftyvar > \neginftyvar$.
  There exists some $k$ such that the $k$-th component of 
  $\sem{s}^\sharp \inftyvar$
  is greater than $\neginfty$.
  We choose $k$ non-deterministically.
  %\tom{Muss noch definiert werden}
  There exists a $\vee$-strategy $\sigma$ for $s$ such that 
  the $k$-th component of
  $\sem{s \sigma}^\sharp \inftyvar$
  equals the $k$-th component of 
  $\sem{s}^\sharp \inftyvar$.
  We choose such a $\vee$-strategy non-deterministically.
%  in polynomial time.
  By Lemma \ref{l:merge-simple:poly},
  we can check in polynomial time, 
  whether the $k$-th component of
  $\sem{s \sigma}^\sharp \inftyvar$
  is greater than $\neginfty$.
  If this is fulfilled,
  we accept.
  
  %\tom{PERFEKT}
  In order to show $\mathrm{NP}$-hardness,
  we reduce the NP-hard problem SAT to our problem.
  Let 
  $\Phi$ 
  be a propositional formula with $n$ variables.
  W.l.o.g.\ we assume that $\Phi$ is in normal form,
  i.e.,
  there are no negated sub-formulas that contain $\wedge$ or $\vee$.
  We define the statement $s(\Phi)$ that uses the variables of $\Phi$ as program variables inductively by
%  \begin{align*}
%    s(z) &:= z = 1, & 
%    s(\overline z) &:= z = 0, \\
%    s(\Phi_1 \wedge \Phi_2) &:= s(\Phi_1) ; s(\Phi_2), &
%    s(\Phi_1 \vee \Phi_2) &:= s(\Phi_1) \mid s(\Phi_2),
%  \end{align*}
    $s(z) := z = 1$, 
    $s(\overline z) := z = 0$,
    $s(\Phi_1 \wedge \Phi_2) := s(\Phi_1) ; s(\Phi_2)$, and 
    $s(\Phi_1 \vee \Phi_2) := s(\Phi_1) \mid s(\Phi_2)$,
  where $z$ is a variable of $\Phi$, and $\Phi_1, \Phi_2$ are formulas.
  Here, the statement $Ax = b$ is an abbreviation for 
  the statement 
  $Ax \leq b; -Ax \leq -b$.
  The formula
  $\Phi$
  is satisfiable iff 
  $\sem{s(\Phi)} \R^{n} \neq \emptyset$ holds.
  Moreover,
  even if we just use the interval domain,
  $\sem{s(\Phi)}\R^{n} \neq \emptyset$ 
  holds iff
  $\sem{s(\Phi)}^\sharp \inftyvar > \neginftyvar$ holds. 
  Thus,
  $\Phi$
  is satisfiable iff 
  $\sem{s(\Phi)}^\sharp \inftyvar > \neginftyvar$ holds. 
  \qed
\end{proof}

%\tom{PERFEKT}
\noindent
Obviously, 
$\sem{(s_1 \mid s_2) ; s} = \sem{s_1; s \mid s_2 ; s}$
and 
$\sem{s ; (s_1 \mid s_2)} = \sem{s; s_1 \mid s ; s_2}$
% hold 
for all statements $s, s_1, s_2$.
%Thus 
We can transform any statement $s$
into an equivalent merge-simple statement $s'$ using these rules.
%\tom{TODO: Besser beschreiben?}
% This can be done in a canonical way.
We denote the merge-simple statement $s'$ that is obtained 
from an arbitrary statement $s$
by applying the above rules in some canonical way by $[s]$.
Intuitively, 
$[s]$ is an explicit enumeration of all paths through the statement $s$.

%\begin{example}
%  %\tom{PERFEKT}
%  For sequential statements $s_1,s_2,s_3,s_4,s_5$,
%  we have
%  $
%    [(s_1 \mid s_2 \mid s_3) ; (s_4 \mid s_5)]
%    =
%    s_1 ; s_4 \mid s_1 ; s_5 \mid 
%    s_2 ; s_4 \mid s_2 ; s_5 \mid 
%    s_3 ; s_4 \mid s_3 ; s_5
%  $.
%  The transformation $[\cdot]$ has explicitly enumerated all paths through 
%  the statement $(s_1 \mid s_2 \mid s_3) ; (s_4 \mid s_5)$.
%  \qed
%\end{example}

\begin{lemma}
  \label{l:into:merge-simple}
  %\tom{PERFEKT}
  For every statement $s$, 
  $[s]$ is merge-simple, and $\sem s = \sem{[s]}$.
  The size of $[s]$ is at most exponential in the size of $s$.
  \qed
\end{lemma}

%\tom{OK}
\noindent
However, in the worst case, 
the size of $[s]$ is
exponential in the size of $s$.
For the statement 
$s = (s_1^{(1)} \mid s_1^{(2)}) ; \cdots ; (s_k^{(1)} \mid s_k^{(2)})$ ,
for instance,
we get 
$
  [s] = |_{(a_1,\ldots,a_k) \in \{1,2\}^k} \allowbreak s_1^{(a_1)} ; \cdots ; s_k^{(a_k)}
  .
$
After replacing all statements $s$ with $[s]$
it is in principle possible to use the methods of \citet{DBLP:conf/csl/GawlitzaS07}
in order to compute the abstract semantics $\Values^\sharp$ precisely.
Because of the exponential blowup, however, 
this method would be impractical in most cases.
\footnote{
Note that we cannot expect a polynomial-time algorithm,
because of Lemma \ref{l:diamand:is:np:complete}: even without
loops, abstract reachability is NP-hard.
Even if all statements are merge-simple,
we cannot expect a polynomial-time algorithm,
since the problem of computing the winning regions 
of parity games %, for instance,
is polynomial-time reducible to abstract reachability
\cite{DBLP:conf/fm/GawlitzaS08}.}

%\tom{OK}
Our new method that we are going to present % in the present article 
avoids this exponential blowup:
instead of enumerating all program paths, we shall visit them only as needed.
Guided by a SAT modulo real linear arithmetic solver, 
our method selects a path through $s$ only when it is \emph{locally profitable} in some sense.
In the worst case, an exponential number of paths may be visited
(Section \ref{s:upper}); but one can hope that this does not happen in many
practical cases, in the same way that SAT and SMT solving perform well on
many practical cases even though they in principle may visit an exponential number
of cases.

\subsection{Abstract Semantic Equations}

% PERFEKT
The first step of our method consists of rewriting our program analysis problem 
into a \emph{system of abstract semantic equations} that is interpreted over the reals. 
For that,
let $G = (N,E,\start)$ be an affine program
and $\Values^\sharp$ its abstract semantics.
We define the system $\C(G)$ of \emph{abstract semantic inequalities} to be the 
smallest set of inequalities that fulfills the following constraints:
\begin{compactitem}
  \item
    $\C$ contains the inequality
    $\vx_{\start,i} \geq \alpha_{i\cdot}(\R^n)$
    for every $i \in \{1,\ldots,m\}$.
  \item
    $\C$ contains the inequality 
    $\vx_{v,i} \geq \sem{s}^\sharp_{i\cdot}  (\vx_{u,1},\ldots,\vx_{u,m})$
    for every control-flow edge $(u,s,v) \in E$ and every $i \in \{1,\ldots,m\}$.
\end{compactitem}

\noindent
We define the system $\E(G)$ of 
\emph{abstract semantic equations} by
$
  \E(G) := \E(\C(G))
$.
Here, for a system 
$ %\begin{align*}
  \C' = \{ \vx_1 \geq e_{1,1}, \ldots, \vx_1 \geq e_{1,k_1},
   \ldots,  
    \vx_n \geq e_{n,1}, \ldots, \vx_n \geq e_{n,k_n} \}
$ %\end{align*}
of inequalities, 
$\E(\C')$ is the system 
$ %\begin{align*}
  \E(\C') = \{ \vx_1 = e_{1,1} \vee\cdots\vee e_{1,k_1}, \ldots,  \vx_n = e_{n,1} \vee\cdots\vee e_{n,k_n} \}
$ %\end{align*}
of equations.
The system $\E(G)$ of abstract semantic equations captures the abstract semantics $\Values^\sharp$ of $G$:

\begin{lemma}
  %\tom{PERFEKT}
  \label{l:eqs:abs_sem}
  $(\Values^\sharp[v])_{i \cdot} = \mu\sem{\E(G)}(\vx_{v,i})$
  for all program points $v$,
  % and all
   $i \in \{ 1,\ldots,m \}$.
  \qed
\end{lemma}

\begin{example}[Abstract Semantic Equations]
  %\tom{Beispiel: PERFEKT}
  \label{ex:running:1}
  We again consider the program $G$ of Example \ref{ex:running:0}.
  Assume that the template constraint matrix $T \in \R^{2 \times 2}$ is given by
  $T_{1 \cdot} = (1, 0)$
  and 
  $T_{2 \cdot} = (-1, 0)$.
  %
  %\noindent
  Let $\Values^\sharp$ denote the abstract semantics of $G$.
  Then $\Values^\sharp[1] = (2001,2000)^\top$.
  $\E(G)$ consists of the following abstract semantic equations:
  \begin{align*}
    \vx_{\start,1} &= \infty &
    \vx_{1,1} &= \sem{x_1 := 0}^\sharp_{1\cdot} (\vx_{\start,1},\vx_{\start,2}) \vee \sem{s}^\sharp_{1\cdot} (\vx_{1,1},\vx_{1,2}) \\[-1mm]
    \vx_{\start,2} &= \infty &
    \vx_{1,2} &= \sem{x_1 := 0}^\sharp_{2\cdot} (\vx_{\start,1},\vx_{\start,2}) \vee \sem{s}^\sharp_{2\cdot} (\vx_{1,1},\vx_{1,2}) 
  \end{align*}
  
  \noindent
  As stated by Lemma \ref{l:eqs:abs_sem},
  we have
  $(\Values^\sharp[1])_{1\cdot} = \mu\sem{\E(G)}(\vx_{1,1}) = 2001$, and 
    $(\Values^\sharp[1])_{2\cdot} = \mu\sem{\E(G)}(\vx_{1,2}) = 2000$.
    \qed
\end{example}

\section{A Lower Bound on the Complexity}
\label{s:lower}

% PERFEKT
In this section we show that the problem of computing abstract semantics of affine programs
w.r.t.\ the interval domain is $\Pi^p_2$-hard.
$\Pi^p_2$-hard problems are conjectured to be harder than both $\mathrm{NP}$-complete and co-$\mathrm{NP}$-complete problems.
For further information regarding the polynomial-time hierarchy see
e.g.\ \citet{Stockmeyer76}.

\begin{theorem}
  \label{t:lower}
  %\tom{PERFEKT}
  The problem of deciding,
  whether,
  for a given program $G$, 
  a given template constraint matrix $T$,
  and a given program point $v$,
  $\Values^\sharp[v] > \neginftyvar$ holds,
  is $\Pi^p_2$-hard.
\end{theorem}

\begin{proof}
  %\tom{PERFEKT}
  We reduce the $\Pi^p_2$-complete problem of
  deciding the truth of a $\forall\exists$ propositional formula
  \citep{DBLP:journals/tcs/Wrathall76}
  to our problem. 
  Let 
  $\Phi = \forall x_1 , \ldots  , x_n . \exists y_1, \ldots, y_m .\Phi'$
  be a formula without free variables, 
  where $\Phi'$ is a propositional formula.
  We consider the affine program $G = (N,E,\start)$, 
  with program variables 
  $x,x', \allowbreak x_1, \ldots, x_n, \allowbreak
   y_1, \ldots, y_m$,
  where
  $N = \{ \start, 1, 2 \}$, 
  and
  $E = \{ (\start, x := 0, 1), \allowbreak (1,s,1), \allowbreak (1,x \geq 2^{n},2) \}$
  with
  \begin{align*}
    s
    \;=\; %&\;
    x' := x; 
    %\\
    &
    \;( x' \geq 2^{n-1}; x' := x' - 2^{n-1}; x_n := 1 \mid x' \leq 2^{n - 1} - 1; x_n := 0 ); \cdots \\
    &\; 
    ( x' \geq 2^{1-1}; x' := x' - 2^{1-1}; x_1 := 1  \mid x' \leq 2^{1-1} - 1; x_1 := 0 ); \\
    &\; 
    s(\Phi') ; \;
    x := x + 1 
  \end{align*}
  The statement $s(\Phi')$ is defined as in the proof of Lemma \ref{l:diamand:is:np:complete}.

  In intuitive terms: this program initializes $x$ to $0$. Then, it enters a loop: 
  it computes into $x_1,\dots,x_n$ the binary decomposition of $x$, 
  then it attempts to nondeterministically choose $y_1,\dots,y_m$ so that $\phi'$ is true. 
  If this is possible, it increments $x$ by one and loops.
  Otherwise, it just loops.
  Thus, there is a terminating computations iff $\Phi$ holds.
    
  Then 
  $\Phi$ 
  holds iff
  $\Values[2] \neq \emptyset$.
  For the abstraction,
  we consider the interval domain.
  By considering the Kleene-Iteration,
  it is easy to see that 
  $\Values[2] \neq \emptyset$ holds iff 
  $\Values^\sharp[2] > \neginftyvar$ holds.
  Thus $\Phi$ holds iff $\Values^\sharp[2] > \neginftyvar$ holds.
  \qed
\end{proof}

\section{Determining Improved Strategies} % by Solving SAT Modulo Real Linear Arithmetic Formulas}
\label{s:improve}

In this section we develop a method for computing local improvements of strategies through solving SAT modulo real linear arithmetic formulas.

%\tom{PERFEKT}
In order to decide,
whether or not,
for a given statement $s$,
a given $j \in \{1,\ldots,m\}$, 
a given $c$, 
and a given $d \in \CR^m$,
$\sem{s}^\sharp_{j\cdot} d > c$ 
holds,
we construct the following SAT modulo real linear arithmetic formula
(we use existential quantifiers %in order 
 to improve readability):
\begin{align*}
  \Phi(s,d,j,c)
  &:\equiv
  \exists v \in \R \;.\; \Phi(s,d,j) \wedge v > c
\\
  \Phi(s,d,j)
  &:\equiv
  \exists
  x \in \R^n,
  x' \in \R^n
  \;.\;
  T x \leq d \wedge \Phi(s) \wedge v = T_{j\cdot} x' %
\end{align*}

%\tom{PERFEKT}
\noindent
Here, $\Phi(s)$ is a formula that relates every $x \in \R^n$ with all elements from the set $\sem{s} \{ x \}$.
%We assume that every occurrence of the operator $\mid$
%is labeled with a unique id $i \in \{1,\ldots,k\}$,
%where $k$ is the number of occurrences of the operator $\mid$.
%We denote by $\mid_i$ the fact that the occurrence is
%labeled with the id $i$.
It is defined inductively over the structure of $s$ as follows:
\begin{align*}
  \Phi(x := Ax + b) 
    &:\equiv x' = Ax + b
    \\
  \Phi(Ax \leq b) 
    &:\equiv Ax \leq b \wedge x' = x
    \\
  \Phi(s_1 ; s_2)
    &:\equiv \exists x'' \in \R^n \; . \; \Phi(s_1)[x''/x'] \wedge \Phi(s_2)[x''/x]
    \\
  \Phi(s_1 \mid s_2)
    &:\equiv (\overline {a_{\pos(s_1 \mid s_2)}} \wedge \Phi(s_1)) 
%    \\&\qquad
        \vee 
        ( a_{\pos(s_1 \mid s_2)} \wedge \Phi(s_2))
\end{align*}

\noindent
Here,
for every position $p$ of a subexpression of $s$,
$a_{p}$ is a Boolean variable.
Let $\Pos_\mid (s)$ denote the set of all positions of $\mid$-sub\-ex\-pres\-sions of $s$.
The set of free variables of the formula 
$\Phi(s)$ is $\{ x, x' \} \cup \{a_p \mid p \in \Pos_\mid(s) \}$.
A valuation for the variables from the set  $\{a_p \mid p \in \Pos_\mid(s) \}$
describes a path through $s$.
We have:

\begin{lemma}
  \label{l:smt}
  %\tom{PERFEKT}
  $\sem{s}^\sharp_{j\cdot} d > c$ holds
  iff
  $\Phi(s,d,j,c)$
  is satisfiable.
  \qed
\end{lemma}

%\tom{PERFEKT}
\noindent
Our next goal is to compute a $\vee$-strategy $\sigma$ for $s$ such 
that $\sem{s\sigma}^\sharp_{j\cdot} d > c$ holds,
provided that $\sem{s}^\sharp_{j\cdot} d > c$ holds.
Let $s$ be a statement, $d \in \CR^m$, $j \in \{1,\ldots,m\}$, and $c \in \R$.
Assume that $\sem{s}^\sharp_{j\cdot} d > c$ holds.
By Lemma \ref{l:smt},
there exists a model 
$M$ of $\Phi(s,d,j,c)$.
We define the $\vee$-strategy $\sigma_M$ for $s$ by
$ %\begin{align*}
   \sigma_M(p) := M(a_p)$  
   for all  
   $p \in \Pos_\mid(s)
$. %\end{align*}
By again applying Lemma \ref{l:smt},
we get
$\sem{s \sigma}^\sharp_{j\cdot} d > c$.
Summarizing we have:

\begin{lemma}
  \label{l:smt:2}
  %\tom{PERFEKT}
  By solving the SAT modulo real linear arithmetic formula $\Phi(s,d,j,c)$ 
  that can be obtained from $s$ in linear time,
  we can decide, whether or not $\sem{s}^\sharp_{j\cdot} d > c$ holds.
  From a model $M$ of this formula,
  we can obtain a $\vee$-strategy $\sigma_M$ for $s$ 
  such that $\sem{s\sigma_M}^\sharp_{j\cdot} d > c$ holds
  in linear time.
  \qed
\end{lemma}

  \begin{figure}
%  \vspaces{-8mm}
  \scalebox{0.93}{
  $
  \begin{array}{@{}r@{\;}l@{}}
    \Phi(s,(0,0)^\top,1,0)
     & \equiv
     \exists v \in \R \;.\; \Phi(s,(0,0)^\top,1) \wedge v > 0
   \\
   \Phi(s,(0,0)^\top,1)
   & \equiv
   \exists x \in \R^2, x' \in \R^2 \;.\; x_{1\cdot} \leq 0 \wedge -x_{1\cdot} \leq 0 \wedge \Phi(s) \wedge v = x'_{1\cdot}
   \\ 
   \Phi(s') 
   &\equiv 
   \exists x'' \in \R^2 \;.\; x_{1\cdot} \leq 1000 
      \wedge x''_{1\cdot} = x_{1\cdot} \wedge x''_{2\cdot} = x_{2\cdot}
      \wedge x'_{1\cdot} = x''_{1\cdot} \wedge x'_{2\cdot} = - x''_{1\cdot}
   \\
   &\equiv 
     x_{1\cdot} \leq 1000 \wedge x'_{1\cdot} = x_{1\cdot} \wedge x'_{2\cdot} = - x_{1\cdot}
   \\
   \Phi(s_1) 
   &\equiv 
   \exists x'' \in \R^2 \;.\; x_{2\cdot} \leq -1 
     \wedge x''_{1\cdot} = x_{1\cdot} \wedge x''_{2\cdot} = x_{2\cdot}
     \wedge x'_{1\cdot} = -2 x''_{1\cdot} \wedge x'_{2\cdot} = x''_{2\cdot}
   \\&\equiv
     x_{2\cdot} \leq -1 \wedge x'_{1\cdot} = -2 x_{1\cdot} \wedge x'_{2\cdot} = x_{2\cdot}
   \\
   \Phi(s_2) 
   &\equiv 
   \exists x'' \in \R^2 \;.\; - x_{2\cdot} \leq 0 
     \wedge x''_{1\cdot} = x_{1\cdot} \wedge x''_{2\cdot} = x_{2\cdot}
     \wedge x'_{1\cdot} = - x''_{1\cdot} + 1 \wedge x'_{2\cdot} = x''_{2\cdot}
   \\
   &\equiv 
     x_{2\cdot} \leq 0 \wedge x'_{1\cdot} = - x_{1\cdot} + 1 \wedge x'_{2\cdot} = x_{2\cdot}
   \\
   \Phi(s_1 \mid s_2)
   &\equiv
   (\overline{a_{1}} \wedge \Phi(s_1)) \vee ( a_{1} \wedge \Phi(s_2))
   %\\&
   \equiv
   (\overline{a_{1}} \wedge x_{2\cdot} \leq -1 \wedge x'_{1\cdot} = -2 x_{1\cdot} \wedge x'_{2\cdot} = x_{2\cdot}) 
   \\&
   \hspace*{43mm}\vee 
   ( a_{1} \wedge x_{2\cdot} \leq 0 \wedge x'_{1\cdot} = - x_{1\cdot} + 1 \wedge x'_{2\cdot} = x_{2\cdot})
   \\
   \Phi(s)
   & \equiv
   \exists x'' \in \R^2 \;.\; \Phi(s')[x''/x'] \wedge \Phi(s_1 \mid s_2)[x''/x]
   \\&\equiv
   x_{1\cdot} \leq 1000 
   \wedge 
   (
   (\overline{a_{1}} \wedge - x_{1\cdot} \leq -1 \wedge x'_{1\cdot} = -2 x_{1\cdot} \wedge x'_{2\cdot} = - x_{1\cdot}) 
   \\&\hspace*{21mm}
   \vee 
   ( a_{1} \wedge - x_{1\cdot} \leq 0 \wedge x'_{1\cdot} = - x_{1\cdot} + 1 \wedge x'_{2\cdot} = - x_{1\cdot})
   )
  \end{array}
  $}
  \vspaces{1mm}
  \caption{Formula for Example \ref{ex:running:2}}
  \vspaces{2mm}
  \label{f:running:2}
  \end{figure}

\begin{example}
  \label{ex:running:2}
  We again continue Example \ref{ex:running:0} and \ref{ex:running:1}.
  We want to know,
  whether %or not
  $\sem{s}^\sharp_{1\cdot} (0,0)^\top \allowbreak > 0$
  holds.
  For that we compute a model of the 
  formula $\Phi(s,(0,0)^\top,1,0)$ which is written down in Figure \ref{f:running:2}.
  $M = \{ a_1 \mapsto 1 \}$ is a model of the formula $\Phi(s,(0,0)^\top, 1,0)$.
  Thus, %$\sigma_M = \{ 1 \mapsto 1 \}$ and 
  we have
  $ %\begin{align*}
    0 < \sem{s \sigma_M}^\sharp_{1\cdot} (0,0)^\top = \sem{s' ; s_2}^\sharp_{1\cdot} (0,0)^\top
  $ %\end{align*}
  by Lemma \ref{l:smt:2}.
  \qed
\end{example}

%\tom{OK}
\noindent
It remains to compute a model 
of $\Phi(s,d,j,c)$.
Most of the state-of-the-art SMT solvers, 
as for instance Yices \cite{yices,DBLP:conf/cav/DutertreM06},
support the computation of models directly; if unsupported, one can compute the model using 
% the 
standard self-reduction techniques. % for SAT.
\begin{comment}
Nevertheless,
we also can compute the model using the standard self-reduction techniques for SAT.
For that, 
assuming that $k$ denotes the number of occurrences of $|$-subexpressions in $s$,
we have to solve at most  $2k - 1$ 
SAT modulo real linear arithmetic queries
each of which can be constructed in linear time.
This works as follows:
If we have fixed the values
for the variables in some subset $A \subseteq \{ a_{p} \mid p \in \Pos_\mid(s) \}$ and the formula
is still satisfiable,
then we fix the value for some variable $a \in \{ a_p \mid p \in \Pos_\mid(s) \} \setminus A$.
Once we assume that $a$ is false.
Once we assume that $a$ is true.
At least one of these cases must result in a satisfiable formula.
Thus we get a possible valuation for the variables 
$A \cup \{ a \}$ by solving at most $2$ SAT modulo real linear arithmetic queries.
Thus, after at most $2k - 1$ SAT modulo real linear arithmetic queries,
we know a satisfying valuation for 
all free variables --- provided that it exists.
\end{comment}

The semantic equations we are concerned with in the present article 
have the form 
$ %\begin{align*}
  \vx = e_1 \vee \cdots \vee e_k,
$ %\end{align*} 
%
%\noindent
where each expression $e_i$, $i = 1,\ldots,k$ is either a constant or an expression of the form 
$\sem{s}^\sharp_{j\cdot}(\vx_1,\ldots,\vx_m)$.
%For simplicity
%we assume w.l.o.g.\ that all expressions $e_i$, $i = 1,\ldots,k$ are of the form
%$\sem{s}^\sharp_{j\cdot}(\vx_1,\ldots,\vx_m)$.
%This can be done w.l.o.g.,
%since for every constant $c \in \R \cup \{\neginfty\}$,
%we can chose some statement $s$ and some of size $\mathcal{O}(1)$ such that 
%$c = \sem{s}^\sharp_{j\cdot} d$ holds for all $d \in $.
%
We now extent our notion of $\vee$-strategies in order to deal with the occurring right-hand sides:

\begin{definition}[$\vee$-Strategies]
\label{d:max:strat}
  The $\vee$-strategy for all constants is the $0$-tuple $()$.
  The application $c()$ of $()$ to a constant $c \in \CR$ is defined by $c() := c$ for all $c \in \CR$.
  A $\vee$-strategy $\sigma$ for an expression $\sem{s}^\sharp_{j\cdot}(\vx_1,\ldots,\vx_m)$ is a $\vee$-strategy for $s$.
  The application $(\sem{s}^\sharp_{j\cdot}(\vx_1,\ldots,\vx_m)) \sigma$ of $\sigma$ to
  $\sem{s}^\sharp_{j\cdot}(\vx_1,\ldots,\vx_m)$ is defined by 
  $ %\begin{align*}
    (\sem{s}^\sharp_{j\cdot}(\vx_1,\ldots,\vx_m)) \sigma := \sem{s \sigma}^\sharp_{j\cdot}(\vx_1,\ldots,\vx_m)
    .
  $ %\end{align*}
  %
  %\noindent
  A $\vee$-strategy for an expression 
  % of the form 
  $ %\begin{align*}
    e = e_0 \vee e_1,
  $, %\end{align*}
  where, for each $i \in \{0,1\}$, $e_i$ is either a constant or an expression of the form $\sem{s}^\sharp_{j\cdot}(\vx_1,\ldots,\vx_m)$,
  is a pair $(p,\sigma)$,
  where $p \in \{0,1 \}$ and $\sigma$ is a $\vee$-strategy for $e_p$. 
  The application $e (p,\sigma)$ of $(p,\sigma)$ to $e = e_0 \vee e_1$ is defined by
$ %  \begin{align*}
    e (p,\sigma) = e_p \sigma
$. %  \end{align*}
  A $\vee$-strategy $\sigma$ for a system $\E = \{ \vx_1 = e_1, \ldots, \vx_n = e_n \}$
  of abstract semantic equations is a mapping 
$ %  \begin{align*}
    \{ \vx_i \mapsto \sigma_i \mid i = 1,\ldots,n\},
$ %  \end{align*}
  where $\sigma_i$ is a $\vee$-strategy for $e_i$ for all $i = 1,\ldots,n$.
  We set $\E(\sigma) := \{ \vx_1 = e_1 (\sigma(\vx_1)), \ldots, \vx_n = e_n (\sigma(\vx_n)) \}$.
  \qed
\end{definition}

%\tom{PERFEKT}
\noindent
Using the same ideas as above,
we can prove the following lemma
which finally enables us to use a SAT modulo real linear arithmetic solver 
for improving $\vee$-strategies for systems of abstract semantic equations locally.

\begin{lemma}
  \label{l:smt:3}
  %\tom{PERFEKT}
  Let $\vx = e$ be an abstract semantic equation,
  $\rho$ a variable assignment, 
  and $c \in \CR$.
  By solving a SAT modulo real linear arithmetic formula 
  that can be obtained from $e$, $\rho$ and $c$ in linear time,
  we can decide, whether or not $\sem{e} \rho > c$ holds.
  From a model $M$ of this formula,
  we can in linear time obtain a $\vee$-strategy $\sigma_M$ for $e$ 
  such that $\sem{e \sigma_M} \rho > c$ holds.
  \qed
\end{lemma}

\section{Solving Systems of Concave Equations}
\label{s:stratimp}

%\tom{PERFEKT}
%As discussed in Section \ref{s:basics},
%we could rewrite all statements of an affine program into equivalent merge-simple statements and
%then apply the methods of \citet{DBLP:conf/csl/GawlitzaS07}.
%This approach is obviously impractical.
%We want to enumerate paths through statements in a more demand-driven fashion.
In order to solve systems of abstract semantic equations (see the end of Section \ref{s:basics})
we generalize the $\vee$-strategy improvement algorithm of \citet{gawlitza_sas_10} as follows:

%\subsection{Improvement Selection Property}

\subsection{Concave Functions}

A set $X \subseteq \R^n$ is called 
%\emph{order-convex}
%iff
%\begin{align*}
%  \lambda x + (1-\lambda) y \in X
%\end{align*}
%holds for all comparable $x,y \in X$ and all $\lambda \in [0,1]$.
%It is called 
\emph{convex} iff 
$ %\begin{align*}
  \lambda x + (1-\lambda) y \in X
$ %\end{align*}
 holds 
for all $x,y \in X$ and all $\lambda \in [0,1]$.
%Thus, any convex set is order-convex, but not vice versa.
%A mapping $f : X \to \R^m$ with $X \subseteq \R^n$ order-convex
%is called \emph{order-convex} (resp.\ \emph{order-concave}) 
%iff
%\begin{align*}
%  f(\lambda x + (1-\lambda)y) 
%  \leq \text{(resp.\ $\geq$)}\;
%  \lambda f(x) + (1-\lambda) f(y)
%\end{align*}
%holds 
%for all comparable $x,y \in M$ and all $\lambda \in [0,1]$ 
%(cf.\ \citet{OrtegaRheinboldt:book}).
A mapping $f : X \to \R^m$ with $X \subseteq \R^n$ convex
is called \emph{convex} (resp.\ \emph{concave}) 
iff
$ %\begin{align*}
  f(\lambda x + (1-\lambda)y) 
  \leq \text{(resp.\ $\geq$)}\;
  \lambda f(x) + (1-\lambda) f(y)
$ %\end{align*}
holds 
for all $x,y \in X$ and all $\lambda \in [0,1]$. 
% (cf.\ e.g.\ \citet{OrtegaRheinboldt:book}).
% Every convex (resp.\ concave) function is order-convex (resp.\ order-concave), but not vice versa.
Note that $f$ is concave iff $-f$ is convex.
Note also that $f$ is convex (resp.\ concave)
iff $f_{i\cdot}$ is convex (resp.\ concave)
for all $i = 1,\ldots,m$.

We extend the notion of convexity/concavity from $\R^n \to \R^m$ to
$\CR^n \to \CR^m$ as follows:
Let $f : \CR^n \to \CR^m$, and $I : \{ 1,\ldots,n \} \to \{ \neginfty, \mathsf{id}, \infty \}$. %, and $k = f^{-1}(\{ 0 \})$,
Here, $\neginfty$ denotes the function that assigns $\neginfty$ to every argument,
$\mathsf{id}$ denotes the identity function,
and $\infty$ denotes the function that assigns $\infty$ to every argument.
We define the mapping $f^{(I)} : \CR^{n} \to \CR^m$ by 
$ %\begin{align*}
  f^{(I)} (x_1,\ldots,x_{n}) := f(I(1)(x_1),\ldots,I(n)(x_n)) 
$
for all $x_1,\ldots,x_{n} \in \CR$.
  A mapping $f : \CR^n \to \CR^m$ is called \emph{concave}
  iff
  $f_{i\cdot}$ is continuous on $\{ x \in \CR^n \mid f_{i\cdot}(x) > \neginfty \}$ for all $i \in \{1,\ldots,m\}$, and
  the following conditions are fulfilled for all $I : \{ 1,\ldots,n \} \to \{ \neginfty, \mathsf{id}, \infty \}$:
  \begin{enumerate}
    \item
      $\fdom(f^{(I)})$ is convex.
    \item
      $f^{(I)}|_{\fdom(f^{(I)})}$ is concave.
    \item
      For all $i \in \{ 1,\ldots,m \}$ the following holds:
      If there exists some $y \in \R^n$ such that $f^{(I)}_{i\cdot}(y) \in \R$,
      then
      $f^{(I)}_{i\cdot}(x) < \infty$ for all $x \in \R^n$.
\end{enumerate}

\noindent
  A mapping $f : \CR^n \to \CR^m$ is called 
  \emph{convex} 
  iff
  $-f$ is concave.
% Note that $\fdom(f)$ is order-convex for all monotone mappings $f : \CR^n \to \CR^m$.
% Thus, since $f^{(I)}$ is monotone for all $I$, $\fdom(f^{(I)})$ is order-convex for all monotone mappings 
% $f : \CR^n \to \CR^m$ and all $I$.
In the following we are only concerned with mappings $f : \CR^n \to \CR^m$
that are monotone and concave.

We slightly extend the definition of concave equations of \citet{gawlitza_sas_10}:

\begin{definition}[Concave Equations]
  %\tom{PERFEKT}
  An expression $e$ (resp.\ equation $\vx = e$) over $\CR$
  is called 
  \emph{basic concave expression} 
  (resp.\ \emph{basic concave equation})
  iff
  $\sem{e}$ is monotone and concave.
%  and $\sem{e}\rho < \infty$ holds
%  for all $\rho : \vX \to \R$,
%  whenever $\fdom(\sem e) \neq \emptyset$.
  %  
  An expression $e$ (resp.\ equation $\vx = e$) over $\CR$
  is called \emph{concave}
  iff $e = \bigvee E$,
  where $E$ is a set of
  basic concave expressions.
%  Thereby we consider the elements from $\CQ$ as nullary operators.
%  
%  An expression (resp.\ equation) over $\CR$ 
%  that only uses operators from the set 
%  $\CQ \, \cup \{ \vee, \wedge, + \} \cup \{ c \cdot \mid c \in \Qpp \}$
%  is called rational expression (resp.\ rational equation).
%  Thereby we consider the elements from $\CQ$ as nullary operators.
  \qed
\end{definition}

%\tom{PERFEKT}
\noindent
The class of systems of concave equations
strictly subsumes the class of \emph{systems of rational equations}
and even the class of \emph{systems of rational LP-equations}
as defined by \citet{DBLP:conf/csl/GawlitzaS07,rationals}
(cf. \cite{gawlitza_sas_10}).

For this paper it is important to observe that every system of abstract semantic equations (cf.\ Section \ref{s:basics}) 
is a system of concave equations:
For every statement $s$,
the expression 
$\sem{s}^\sharp_{j\cdot}(\vx_1,\ldots,\vx_m)$
is a concave expression, 
since 
(1) the expression $( \sem{s}^\sharp_{j\cdot}(\vx_1,\ldots,\vx_m) ) \sigma$ is a basic concave expression for all $\vee$-strategies $\sigma$,
(i.e. $\sem{s \sigma}	^\sharp_{j\cdot}$ is monotone and concave)
and 
(2)
the expression $\sem{s}^\sharp_{j\cdot}(\vx_1,\ldots,\vx_m)$ 
can be written as the expression $\bigvee_{\sigma\in\Sigma} (\sem{s}^\sharp_{j\cdot}(\vx_1,\ldots,\vx_m)) \sigma$.
Here, $\Sigma$ denotes the set of all $\vee$-strategies.
Hence, we can generalize the concept of $\vee$-strategies as follows:

% Moreover, every system of \emph{basic} abstract semantic equations is a system of basic concave equations.

\paragraph{Strategies}
%\tom{PERFEKT}
A \emph{$\vee$-strategy $\sigma$} 
for $\E$
is a function that maps every expression 
$\bigvee E$ 
occurring in $\E$
to one of the $e \in E$.
We denote the set of all 
$\vee$-strategies 
for $\E$ by 
$\MaxStrat_\E$.
We drop subscripts, whenever they are clear from the context.
For $\sigma \in \MaxStrat$,
the expression $e \sigma$ denotes the expression $\sigma(e)$.
Finally, we set $\E(\sigma) := \{ \vx = e \sigma \mid \vx = e \in \E \}$.

%\input{../textbausteine/strat_imp}

%%%%%%%%%%%%%%%%%%%%%%%%%%%%%%%%%%%%
%%%%%%%%%%%%%%%%%%%%%%%%%%%%%%%%%%%%
%%%%%%%%%%%%%%%%%%%%%%%%%%%%%%%%%%%%

\subsection{The Strategy Improvement Algorithm}

%\tom{PERFEKT}
We briefly explain the strategy improvement algorithm (cf. \cite{rationals,gawlitza_sas_10}).
%In the following we always  
%consider complete \emph{linear} ordered sets.
%
%\subsection{Strategy Improvement}
%
%\noindent
%Our strategy improvement algorithm 
It
%presented in this article 
iterates over $\vee$-strategies.
It maintains a current $\vee$-strategy
and a current \emph{approximate} to the least solution.
A so-called \emph{strategy improvement operator} 
is used for determining a next, improved $\vee$-strategy.
In our application, 
the strategy improvement operator is realized by a SAT modulo real linear arithmetic solver 
(cf.\ Section \ref{s:improve}).
Whether or not a $\vee$-strategy represents 
an \emph{improvement} may depend on the current approximate.
It can indeed be the case
that a switch from one $\vee$-strategy to another $\vee$-strategy 
is only then \emph{profitable},
when it is known,
that the least solution is of a certain size.
Hence, 
we talk about an \emph{improvement}
of a $\vee$-strategy w.r.t.\ an approximate:

\begin{definition}[Improvements]
%\tom{Def.: PERFEKT}
	\label{d:alg:verbesserung}
	\index{Verbesserung}
	\index{Strategieverbesserungsoperator}
	Let $\E$ be a system of monotone equations
	over a complete linear ordered set.
	Let $\sigma, \sigma' \in \MaxStrat$ be
	$\vee$-strategies for $\E$ 
	and $\rho$ be a pre-solution of $\E(\sigma)$.
	The $\vee$-strategy $\sigma'$ is called
	\emph{improvement of $\sigma$ w.r.t.\ $\rho$}
	iff the following conditions are fulfilled:
%	\begin{enumerate}
%		\item
%            \label{d:alg:verbesserung:1}
     (1)
			If $\rho \notin \Sol(\E)$,
			then $\sem{\E(\sigma')}\rho > \rho$.
%		\item
%			\label{d:alg:verbesserung:2}
(2)
			For all $\bigvee$-expressions 
			$e$ occurring in $\E$
			the following holds:
			If $\sigma'(e) \neq \sigma(e)$,
			then $\sem{e \sigma'} \rho > \sem{e \sigma} \rho$.
%	\end{enumerate}
	A function $\Pv$ which assigns 
	an improvement of $\sigma$ w.r.t.\ $\rho$
	to every pair $(\sigma,\rho)$,
	where $\sigma$ is a $\vee$-strategy 
	and $\rho$ is a pre-solution of $\E(\sigma)$,
	is called 
	\emph{$\vee$-strategy improvement operator}.
%	
%	An \emph{improvement of a $\wedge$-strategy $\pi$ w.r.t.\ a post-solution of $\E(\pi)$}
%	and
%	a \emph{$\wedge$-strategy improvement operator for $\E$}
%	are defined dually.
	\qed
\end{definition}

\noindent
In many cases, there exist several, different improvements of a 
$\vee$-strategy $\sigma$ w.r.t.\ a pre-solution $\rho$ of $\E(\sigma)$.
Accordingly, there exist several, different 
strategy improvement operators.
One possibility for improving the current strategy 
is known as 
\emph{all profitable switches}~\citep{bjoerklund02,Bjork2003}.
Carried over to the case considered here, 
this means:
For the improvement $\sigma'$ of $\sigma$ w.r.t.\ $\rho$ we have:
$\sem{\E(\sigma')} \rho = \sem{\E}\rho$,
i.e., $\sigma'$ represents the best local improvement of $\sigma$ at $\rho$.
We denote $\sigma'$ by $\Pve(\sigma,\rho)$~\citep{DBLP:conf/esop/GawlitzaS07,DBLP:conf/csl/GawlitzaS07,DBLP:conf/fm/GawlitzaS08,rationals}.

Now we can formulate the 
strategy improvement algorithm for
computing least solutions of systems of monotone equations
over complete linear ordered sets. 
This algorithm is parameterized with a 
$\vee$-strategy improvement operator $\Pv$.
The input is a system $\E$ of monotone equations
over a complete linear ordered set,
a $\vee$-strategy $\sigma_{\mathrm{init}}$ for $\E$, 
and
a pre-solution $\rho_{\mathrm{init}}$ 
of $\E(\sigma_{\mathrm{init}})$.
In order to compute the \emph{least} and not some \emph{arbitrary} solution,
we additionally assume that $\rho_{\mathrm{init}} \leq \mu\sem\E$ holds:

\vspaces{-5mm}

%\tom{Alg: PERFEKT}
\begin{algorithm}[H]
	$
	\\[0mm]
	\begin{array}{@{}l@{\;\text{:}\;}l@{}}
%		\text{Parameter} 
%		&
%		\text{A $\vee$-strategy improvement operator $\Pv$} 
%		   \\[1mm]
		\text{Input} 
		&
		\left\{
		\begin{array}{@{}l@{}}
			%\text{- Ein $\vee$-Strategieverbesserungsoperator $\Pv$} \\
			\text{- A system $\E$ of monotone equations over a complete linear ordered set }
			            \\
			\text{- A $\vee$-strategy $\sigma_{\mathrm{init}}$ for $\E$} \\
			\text{- A pre-solution $\rho_{\mathrm{init}}$ 
			           of $\E(\sigma_{\mathrm{init}})$ 
			           with $\rho_{\mathrm{init}} \leq \mu\sem\E$} \\
		\end{array}
		\right.
		\\[0mm]
%		\text{Output} 
%		&
%		\text{The least solution $\mu\sem\E$ of $\E$} 
	\end{array} \\[1mm]
    \sigma \GETS \sigma_{\mathrm{init}} ; \;
    \rho \GETS \rho_{\mathrm{init}} ; \,
%    \VS
	\WHILE (\rho \notin \Sol(\E)) \;
	\{
		\sigma \GETS \Pv(\sigma,\rho) ; \;
		\rho \GETS \mu_{\geq \rho} \sem{\E(\sigma)} ; 
	\} 
	\; 
	\RETURN \rho;
	\\[-3mm]
	$
	\caption{The Strategy Improvement Algorithm}
	\label{alg:alg:stratimp}
\end{algorithm}

\vspaces{-10mm}

%%\tom{PEFEKT}
%\noindent
%For executing the strategy improvement algorithm 
%(algorithm~\ref{alg:alg:stratimp}),
%we need,
%for the particular class of systems of monotone equations,
%a method for computing 
%$\mu_{\geq \rho} \sem{\E(\sigma)}$ 
%for the $\vee$-strategies $\sigma$ 
%and 
%the approximates $\rho$
%occurring during the execution.

%%\tom{PERFEKT}
%We have:

\begin{lemma}
	\label{l:alg:sequence}
	%\tom{Lemma: PERFEKT}
	Let $\E$ be a system of monotone equations
	over a complete linear ordered set.
    For $i \in \N$, 
    let
    $\rho_i$ be the value of the program variable $\rho$ 
    and
    $\sigma_i$ be the value of the program variable $\sigma$ 
    in the strategy improvement algorithm  
%    (algorithm \ref{alg:alg:stratimp})
    after the $i$-th evaluation of the loop-body.
	The following statements hold for all $i \in \N$:
	
%	\begin{enumerate}
%		\item
%            \label{l:alg:sequence:1}
%			$\rho_i \leq \mu\sem\E$.
%		\item
%            \label{l:alg:sequence:2}
%			$\rho_i \in \PreSol(\E(\sigma_{i+1}))$.
%		\item
%            \label{l:alg:sequence:3}
%			If $\rho_i < \mu\sem\E$, then $\rho_{i+1} > \rho_i$.
%		\item
%            \label{l:alg:sequence:4}
%			If $\rho_i = \mu\sem\E$, then $\rho_{i+1} = \rho_i$.
%			\qed
%	\end{enumerate}

	\begin{tabular}{rl@{\qquad}rl}
			1. & $\rho_i \leq \mu\sem\E$. &
			2. & $\rho_i \in \PreSol(\E(\sigma_{i+1}))$. \\
			3. & If $\rho_i < \mu\sem\E$, then $\rho_{i+1} > \rho_i$. &
			4. & If $\rho_i = \mu\sem\E$, then $\rho_{i+1} = \rho_i$.
			\qed
	\end{tabular}
\end{lemma}

\noindent
An immediate consequence of Lemma \ref{l:alg:sequence} is
the following:
Whenever the strategy improvement algorithm  
    terminates, 
    it computes the least solution $\mu\sem\E$ of $\E$.

At first we are interested in solving systems of concave equations with \emph{finitely} many strategies and \emph{finite} least solutions.
We show that our strategy improvement algorithm 
terminates and thus returns the least solution 
in this case at the latest after considering all strategies.
Further, we give an important characterization 
for $\mu_{\geq\rho}\sem{\E(\sigma)}$.

\subsection{Feasibility}

%\tom{PERFEKT}
In order to prove termination %in this case
 we define the following notion of feasibility:

\begin{definition}[Feasibility (\cite{gawlitza_sas_10})]
%\tom{Definition: PERFEKT}
  %
  Let $\E$ be a system of basic concave equations. % with $n$ variables.
  A finite solution $\rho$ of $\E$ is called \emph{($\E$-)feasible}
  iff
  there exists $\vX_1, \vX_2 \subseteq \vX$ and some $k \in \N$ such that the following statements hold:
 
  \begin{compactenum}
    \item
      $\vX_1 \cup \vX_2 = \vX$, and $\vX_1 \cap \vX_2 = \emptyset$.
        \item
          There exists some
          $\rho' \ll \rho|_{\vX_1}$
          such that
          $\rho' \;\dot\cup\; \rho|_{\vX_2}$ is a pre-solution of $\E$,
          and
          $\rho = \sem{\E}^k (\rho' \;\dot\cup\; \rho|_{\vX_2})$.
        \item
          There exists a 
          $\rho' \ll \rho|_{\vX_2}$
          such that
          $\rho' \ll (\sem{\E}^k ( \rho|_{\vX_1} \;\dot\cup\; \rho' ))|_{\vX_2}$.
  \end{compactenum}

  %\tom{PERFEKT}
  \noindent
  A finite pre-solution $\rho$ of $\E$ is called \emph{($\E$-)feasible}
  iff 
  $\mu_{\geq \rho} \sem\E$ is a feasible finite solution of $\E$.
  A pre-solution $\rho \ll \inftyvar$ is called feasible iff 
  $e = \neginfty$ for all $\vx = e \in \E$ with $\sem e \rho = \neginfty$,
  and 
  $\rho |_{\vX'}$ is a feasible finite pre-solution of $\{ \vx = e \in \E \mid \vx \in \vX' \}$,
  where $\vX' := \{ \vx \mid \vx = e \in \E, \sem e \rho > \neginfty \}$.
  
  A system $\E$ of basic concave equations is called
  \emph{feasible} iff there exists a feasible solution $\rho$ of $\E$. 
  %
%  Let $\E = \{ \vx_1 = e_1, \ldots, \vx_n = e_n \}$ be a system of basic order-concave equations.
%  A pre-solution $\rho \ll \inftyvar$ of $\E$ is called \emph{($\E$-)feasible}
%  iff 
%  $\rho^* := \mu_{\geq\rho}\sem{\E} \ll \inftyvar$  
%  and
%  there exist variable assignments
%  $\rho_1 \leq \cdots \leq \rho_n \leq \rho^*$
%  and
%  a bijection $p : \{ 1,\ldots,n \} \to \{ 1,\ldots,n \}$
%  such that the following properties hold:
%  \begin{enumerate}
%    \item
%      $\rho_i(\vx_{p(i)}) < \sem{e_{p(i)}} \rho_i$
%      for all 
%      $i = 1,\ldots,n$.
%    \item
%      $\rho_i(\vx_{p(j)}) \leq \sem{e_{p(j)}} \rho_i$
%      for all 
%      $i = 1,\ldots,n$
%      and all $j = 1,\ldots,i - 1$.      
%    \item
%      For all $i = 1,\ldots,n$,
%      either $\sem{e_i} \rho_n \in \R$ or $e_i = \neginfty$.
%  \end{enumerate}
%  %
%  A system $\E$ of basic order-concave equations is called
%  \emph{feasible} iff there exists a feasible pre-solution $\rho$ of $\E$. 
%  %
  \qed
\end{definition}

%\begin{example}
%%\tom{Beispiel: PERFEKT}
%  We consider the system 
%  $\E = \{ \vx = \sqrt \vx \}$
%  of basic order-concave equations.
%  Note that $\Ext(\sem\E) = \{ \sem\E \}$ holds.
%  For all $x \in \CR$,
%  let $\underline x := \{ \vx \mapsto x \}$.
%  $\underline 0$ is not a feasible pre-solution,
%  since 
%  $
%  \underline x(\vx) 
%  = 
%  x \geq \sqrt{x} 
%  = 
%  \sem{\sqrt{\vx}} \underline x 
%  $
%  for all $x \leq 0 = \mu_{\geq \underline 0}\sem{\E}$.
%  $\underline{x}$ is a feasible pre-solution for all $x \in (0,1]$,
%  since $\underline x(\vx) = x < \sqrt{x} = \sem{\sqrt \vx}\underline x$
%  for all $x \in (0,1]$.
%  \qed
%\end{example}

%\tom{PERFEKT}
\noindent
The following lemmas ensure that our strategy improvement 
algorithm stays in the feasible area,
whenever it is started in the feasible area.

\begin{lemma}[\cite{gawlitza_sas_10}]
  \label{l:feas:greater}
  Let $\E$ be a system of basic concave equations and
  $\rho$ be a feasible pre-solution of $\E$.
  Every pre-solution $\rho'$ of $\E$ with 
  $\rho \leq \rho' \leq \mu_{\geq \rho}\sem\E$ 
  is feasible. 
  \qed
\end{lemma}

\begin{lemma}[\cite{gawlitza_sas_10}]
%\tom{Lemma: PERFEKT}
	\label{l:rat:stratimperhaeltzulaessigkeit}
	Let $\E$ be a system of concave equations, 
	$\sigma$ be a $\vee$-strategy for $\E$, 
	$\rho$ be a feasible 
	%pre-solution 
     solution 
	of $\E(\sigma)$,
	and
	$\sigma'$ be an improvement of $\sigma$ w.r.t.\ $\rho$.
	Then $\rho$ is a feasible pre-solution of $\E(\sigma')$.
%	
%    Thereby, 
%    we assume that there exists some 
%    $f \in \Ext(\sem{ \E' } )$
%    such that there exists some $\rho' \in \fdom(f)$ with 
%    $\rho' \ll \rho|_{\vX_{\E'}}$,
%    where $\E' := \{ \vx = e \in \E(\sigma') \mid e \neq \neginfty \}$.
%    \assumption
  \qed
\end{lemma}

\noindent
In order to start in the feasible area,
we simply start the strategy improvement algorithm 
with the system
$
  \E \vee \neginfty := \{ \vx = e \vee \neginfty \mid \vx = e \in \E \}
  ,
$
a $\vee$-strategy 
$\sigma_\mathrm{init}$ for $\E \vee \neginfty$
such that $(\E \vee \neginfty) (\sigma_\mathrm{init}) = \{ \vx = \neginfty \mid \vx = e \in \E \}$,
and the feasible pre-solution $\neginftyvar$ of
$(\E \vee \neginfty) (\sigma_\mathrm{init})$.

%\tom{PERFEKT}
It remains to determine 
$\mu_{\geq \rho}\sem{\E}$.
Because of Lemma \ref{l:feas:greater} and Lemma \ref{l:rat:stratimperhaeltzulaessigkeit}, we are allowed to assume that
$\rho$ is a feasible pre-solution of the 
system $\E$ of basic concave equations.
This is important in our strategy improvement algorithm.
The following lemma in particular states that we 
have to compute the greatest finite pre-solution. % of $\E$.

\begin{lemma}[\cite{gawlitza_sas_10}]
%\tom{Lemma: PERFEKT}
\label{l:eindeutige:loesung}
  Let $\E$ be a feasible system of basic concave equations with
  $e \neq \neginfty$ for all $\vx = e \in \E$.
  There exists a greatest finite pre-solution $\rho^*$ of $\E$
  and $\rho^*$
  is the only 
  feasible solution of $\E$.
  If $\rho$ is a finite pre-solution of $\E$, 
  then $\rho^* = \mu_{\geq \rho}\sem\E$.
  \qed
\end{lemma}

%%\tom{PERFEKT}
%\noindent
%%Finally, 
%%we get the following main result as a corollary of lemma
%%\ref{t:rat:existence_least_consistent_conj}:
%The following theorem collects the key properties of
%feasible systems of conjunctive rational equations:

%\begin{theorem}[Least Feasible Solutions]
%%\tom{Theorem: OK}
%	\label{t:rat:least_sol_gt_least_cons_sol}
%	Let $\E$ be a feasible system of basic order-concave equations.
%	The following statements hold:
%	\begin{enumerate}
%	  \item
%	    \label{t:rat:least_sol_gt_least_cons_sol:1}
%	    There exists a least feasible solution $\rho^*$ of $\E$.
%	  \item
%	    \label{t:rat:least_sol_gt_least_cons_sol:2}
%        If $\rho\ll\inftyvar$ is a feasible pre-solution of $\E$,
%	    then $\rho^* = \mu_{\geq \rho}\sem\E$.
%	  \item
%	    \label{t:rat:least_sol_gt_least_cons_sol:3}
%	    There exists at most one feasible solution $\rho^*$ of $\E$ 
%	    with $\rho^* \ll \inftyvar$.
%	  \item
%	    \label{t:rat:least_sol_gt_least_cons_sol:4}
%        For every $\wedge$-strategy $\pi \in \MinStrat$ let
%        $\rho_\pi$ denote the least feasible solution of $\E(\pi)$.
%        Then $\rho^* = \min_{\pi \in \MinStrat} \; \rho_\pi$.
%	\end{enumerate}
%\end{theorem}

%\begin{proof}
%  The proof is almost literally the same as the proof for 
%  theorem 3 in \cite{rationals}.
%  \qed
%\end{proof}

\subsection{Termination}

%\tom{PERFEKT}
Lemma \ref{l:eindeutige:loesung} implies that 
our strategy improvement algorithm has to consider 
each $\vee$-strategy at most once.
Thus, we have shown the following theorem:

\begin{theorem}
  \label{t:strat:imp:concave}
  %\tom{PERFEKT}
  Let $\E$ be a system of concave equations with $\mu\sem\E \ll \inftyvar$.
  %H: and assume that assumption \assumption is satisfied. 
  Assume that we can compute the greatest finite pre-solution $\rho_\sigma$
  of each $\E(\sigma)$,
  if $\E(\sigma)$ is feasible.
  Our strategy improvement algorithm computes $\mu\sem\E$ and
  performs at most $|\MaxStrat| + |\vX|$ strategy improvement steps.
  The algorithm in particular terminates,
  whenever $\Sigma$ is finite.
  %H: naechsten Satz streichen!
%  Thereby, we assume that assumption \assumption is always
%  fulfilled.
  \qed
\end{theorem}

%%\tom{TODO}
%If assumption \assumption is not fulfilled at some strategy improvement step,
%then there either does not exist a greatest finite pre-solution $\rho^*$,
%or $\mu_{\geq\rho}\sem{\E(\sigma)} \leq \rho^*$.
%This means that we are at least computing a safe over-approximation,
%when we compute the greatest finite pre-solution $\rho^*$ in
%every strategy improvement step instead of $\mu_{\geq\rho}\sem{\E(\sigma)}$.
%However, 
%we strongly believe that $\mu_{\geq\rho}\sem{\E(\sigma)}$
%is in fact always the greatest finite pre-solution
%for the $\vee$-strategies $\sigma$ and the pre-solutions $\rho$ 
%considered in the algorithm.
%It remains for future work to come with a notion of feasibility 
%such that one can get rid of assumption \assumption.

\section{Computing Greatest Finite Pre-Solutions}
\label{s:gfps}

%\tom{PERFEKT}
For all systems $\E$ of abstract semantic equations (see Section \ref{s:basics})
and all $\vee$-strategies $\sigma$,
$\E(\sigma)$ is a system of abstract semantic equations,
where each right-hand side is of the form 
$\sem{s}^\sharp_{j\cdot} (\vx_1,\ldots,\vx_m)$,
where $s$ is a sequential statement and $\vx_1,\ldots,\vx_m$ are variables.
We call such a system of abstract semantic equations a system of \emph{basic} abstract semantic equations.
It remains to explain how we can compute the greatest finite solution of such a system --- provided that it exists.

%\tom{PERFEKT}
Let 
$
  \E 
$
be a system of basic abstract semantic equations
with 
% $\vx = e$,
%where each right-hand side $e$ is of the form
%$\sem{s}^\sharp_{j\cdot} (\vx_1,\ldots,\vx_m)$,
%where 
%$s$ is a sequential statement and $\vx_1,\ldots,\vx_m$ are variables.
a greatest finite pre-solution $\rho^*$.
We can compute $\rho^*$ through linear programming as follows:

%\tom{PERFEKT}
We assume w.l.o.g.\ that every sequential statement $s$ that occurs in the right-hand sides of $\E$ is
of the form 
$ %\begin{align*}
  Ax \leq b; x := A' x + b',
$ %\end{align*}
where
$
  A \in \R^{k \times n},
  b \in \R^k,
  A' \in \R^{n \times n},
  b' \in \R^n
$.
This can be done w.l.o.g., since every sequential statement can be rewritten into this form in polynomial time. 
We define the system $\C$ of linear inequalities to be the smallest set that fulfills the following properties:
For each equation 
\[ \vx = \sem{Ax \leq b; x := A' x + b'}^\sharp_{j\cdot}(\vx_1,\ldots,\vx_m), \]
%with
%$
%  A \in \R^{k \times n},
%  b \in \R^k,
%  A' \in \R^{n \times n},
%  b' \in \R^n
%$,
the system $\C$ contains the following constraints:
\begin{align*}
  \vx    &\leq T_{j\cdot} A' (\vy_1,\ldots,\vy_n)^\top + T_{j\cdot} b' &
  A_{i\cdot} (\vy_1,\ldots,\vy_n)^\top &\leq b_i \text{ for all } i = 1,\ldots,k \\
  && T_{i\cdot} (\vy_1,\ldots,\vy_n)^\top &\leq \vx_i \text{ for all } i = 1,\ldots,m
\end{align*}

\noindent
Here, $\vy_1,\ldots,\vy_n$ are fresh variables.
Then
$ %\begin{align*}
  \rho^*(\vx) = \sup \; \{ \rho(\vx) \mid \rho \in \Sol(\C) \}
$. %\end{align*}
Thus $\rho^*$ can be determined by solving $\abs{\vX_\E}$ linear programming problems each of which 
can be constructed in linear time.
We can do even better by determining an optimal solution of the linear programming problem 
$ %\begin{align*}
  \textstyle
  \sup \left\{ \sum_{\vx \in \vX_\E} \rho(\vx) \mid \rho \in \Sol(\C) \right\}
  .
$ %\end{align*}
Then the optimal values for the variables $\vx \in \vX_\E$ determine $\rho^*$
(cf. \citet{DBLP:conf/csl/GawlitzaS07,rationals}).
Summarizing we have:

\begin{lemma}
     \label{l:comp:greatest:by:lp}
	Let 
	$
	  \E 
	$
	be a system of basic abstract semantic equations 
%	$\vx = e$,
%	where each right-hand side $e$ is of the form
%	$\sem{s}^\sharp_{j\cdot} (\vx_1,\ldots,\vx_m)$,
%	where 
%	$s$ is a sequential statement, $j \in \{1,\ldots,m\}$, and $\vx_1,\ldots,\vx_m$ are variables.
	with a greatest finite pre-solution $\rho^*$.
	Then $\rho^*$ can be computed by solving a linear programming problem 
	that can be constructed in linear time.
	\qed
\end{lemma}

\begin{example}
  \label{ex:running:3}
  We again use the definitions of Example \ref{ex:running:1}.
  Consider the system $\E$ of basic abstract semantic equations that consists of the equations
  \begin{align*}
    \vx_{1,1} = \sem{s' ; s_2}^\sharp_{1\cdot}(\vx_{1,1},\vx_{1,2})
    \qquad
    \vx_{1,2} = \sem{s' ; s_1}^\sharp_{2\cdot}(\vx_{1,1},\vx_{1,2})
    ,
  \end{align*}
  where
    $s'   :=  x_1 \leq 1000 ; x_2 := -x_1$,
    $s_1 := x_2 \leq -1; x_1 := -2 x_1$, and
    $s_2 := -x_2 \leq 0; x_1 := -x_1+1$.
  Our goal is to compute the greatest finite pre-solution $\rho^*$ of $\E$.
  Firstly, we note that 
  $\sem{s'; s_2} = \sem{x_1 \leq 0; (x_1,x_2) := (-x_1 + 1, -x_1)}$ and
  $\sem{s'; s_1} = \sem{(x_1,-x_1) \leq (1000,-1); (x_1,x_2) := (-2x_1, -x_1)}$
  hold.
  Accordingly,
  we have to find an optimal solution for the following linear programming problem:
  \begin{align*}
    \textrm{maximize} \; &\vx_{1,1} + \vx_{1,2} %\textrm{~subject to}
    \\[-1mm]
%    \begin{array}{r@{\;}c@{\;}l@{\qquad}r@{\;}c@{\;}l@{\qquad}r@{\;}c@{\;}l@{\qquad}r@{\;}c@{\;}l@{\quad}r@{\;}c@{\;}l}
    \vx_{1,1} &\leq -\vy_1 + 1 &
    \vx_{1,2} &\leq 2\vy_1' 
    &
    \vy_1       &\leq 0 &
    \vy_1'       &\leq 1000 
    &
    \vy_1       &\leq \vx_{1,1} 
    \\[-1mm]
    -\vy_1'       &\leq -1 
    &
    -\vy_1      &\leq \vx_{1,2} &
    \vy_1'       &\leq \vx_{1,1} 
    &
    -\vy_1'      &\leq \vx_{1,2} 
%    \end{array}
  \end{align*}

  \noindent
  An optimal solution % of this linear programming problem 
  is
    $\vx_{1,1} = 2001$,
    $\vx_{1,2} = 2000$,
    $\vy_1 = -2000$, and
    $\vy_1' = 1000$.
  Thus
  $\rho^* = \{ \vx_{1,1} \mapsto 2001 ,\; \vx_{1,2} \mapsto 2000 \}$
  is the greatest finite pre-solution of $\E$.
  \qed
\end{example}

\noindent
Summarizing, we have shown our main theorem: 
% for systems of abstract semantic equations:

\begin{theorem}
  \label{t:strat:imp:abs:sem:eqs}
  %\tom{TODO: Pruefen}
  Let $\E$ be a system of abstract semantic equations with 
  $\mu\sem\E \ll \inftyvar$.
  Our strategy improvement algorithm computes $\mu\sem\E$ and
  performs at most $|\MaxStrat| + |\vX|$ strategy improvement steps.
  For each strategy improvement step,
  we have 
  to do the following:
%  to perform the following tasks:
  \begin{compactenum}
    \item
      Find models for $|\vX|$ SAT modulo real linear arithmetic formulas,
      each of which can be constructed in linear time.
    \item
      Solve a linear programming problem which can be constructed in linear time.
  \end{compactenum}
\end{theorem}

\begin{proof}
  The statement 
  follows from 
  Lemmas \ref{l:feas:greater}, 
  \ref{l:rat:stratimperhaeltzulaessigkeit}, 
  \ref{l:eindeutige:loesung}, \ref{l:comp:greatest:by:lp} and
  Theorem \ref{t:strat:imp:concave}.
  \qed
\end{proof}

\noindent
Our techniques can be extended straightforwardly
in order to get rid of the pre-condition $\mu\sem{\E} \ll \inftyvar$.
However, 
for simplicity we eschew these technicalities in the present article.
% Nevertheless,
% we will study an example where $\mu\sem{\E} \ll \inftyvar$ does not hold in the next section.

%\subsection{Discussion}
%
%The last $\vee$-strategy $\sigma_4$ leads to a system of abstract semantic equations,
%where, for program point $1$,
%the paths through $s$ depends on the component of the abstract value.

\section{An Upper Bound on the Complexity}
\label{s:upper}

%\tom{PERFEKT}
In Section \ref{s:lower},
we have provided a lower bound on the complexity of computing abstract semantics 
of affine programs w.r.t. the template linear domains. 
In this section we show that the corresponding decision problem is not only 
$\Pi^p_2$-hard, but in fact $\Pi^p_2$-complete:

\begin{theorem}
\label{t:upper}
  The problem of deciding,
  whether or not, % or not,
  for a given affine program $G$, 
  a given template constraint matrix $T$,
  and a given program point $v$,
  $\Values^\sharp[v] > \neginftyvar$ holds,
  is in $\Pi^p_2$.
\end{theorem}

\begin{proof}{\it (Sketch)}
  %\tom{TODO: Ueberpruefen}
  We have to show that the problem of deciding,
  whether or not,
  for a given affine program $G$,
  a given template constraint matrix $T$,
  a given program point $v$,
  and a given $i \in \{1,\ldots,m \}$,
  $(\Values^\sharp[v])_{i\cdot} = \neginfty$ holds,
  is in 
  $\mathsf{co}{-}\Pi^p_2 = \Sigma^p_2 = \mathrm{NP}^\mathrm{NP}$.
  In polynomial time we can guess a $\vee$-strategy $\sigma$ for 
  $\E' := \E(G)$ and compute the \emph{least feasible solution} $\rho$ 
  of $\E'(\sigma)$ 
  (see \citet{DBLP:conf/csl/GawlitzaS07}).
  Because of Lemma \ref{l:diamand:is:np:complete},
  we can use a NP oracle to determine whether or not there exists an improvement of the strategy $\sigma$ w.r.t.\ $\rho$.
  If this is not the case,
  we know that 
  $\rho \geq \mu\sem{\E'}$ holds.
  Therefore,
  by Lemma \ref{l:eqs:abs_sem},
  we have $\rho(\vx_{v,i}) \geq (\Values^\sharp[v])_{i\cdot}$.
  Thus we can accept,
  whenever $\rho(\vx_{v,i}) = \neginfty$ holds.
  \qed
\end{proof}

%%\tom{PERFEKT}
%\noindent
%Theorem \ref{t:lower} together with theorem \ref{t:upper} gives us:
%  The problem of deciding,
%  whether or not,
%  for a given affine program $G$, 
%  a given template constraint matrix $T$,
%  and a given program point $v$,
%  $\Values^\sharp[v] > \neginftyvar$ holds,
%  is $\Pi^p_2$-complete.

%\tom{PERFEKT}
\noindent
Finally, 
we give an example where our strategy improvement algorithm 
performs exponentially many strategy improvement steps.
It is similar to the program in the proof of Theorem \ref{t:lower}.
For all $n \in \N$,
we consider the program $G_n = (N,E,\start)$,
  where
  $N = \{ \start, 1 \}$, 
  $E = \{ (\start, x_1 := 0 ; y_1 := 1; y_2 := 2 y_1; \ldots ; y_n := 2 y_{n-1} , 1), (1,s,1) \}$,
  and
  \begin{align*}
    s
    \;=\;
    x_2 := x_1; 
    &\;( x_2 \geq y_n; x_2 := x_2 - y_n \mid x_2 \leq y_n - 1 ); \cdots \\
    &\; 
    ( x_2 \geq y_1; x_2 := x_2 - y_1 \mid x_2 \leq y_1 - 1 ); 
    \; 
    x_1 := x_1 + 1 
    .
  \end{align*}
  
\noindent
It is sufficient to use a template constraint matrix that corresponds to the interval domain.
It is remarkable that the strategy iteration does not depend on the
strategy improvement operator in use.
At any time there is exactly one possible improvement until the least solution is reached.
All strategies for the statement $s$ will be encountered.
Thus, the strategy improvement algorithm performs $2^n$ strategy improvement steps.
Since the size %$\abs{G_n}$ 
of $G_n$ is $\Theta(n)$,
exponentially many strategy improvement steps are performed.

\begin{comment}
\section{Experimental Results}
\label{s:experiments}

%\tom{PERFEKT}
We have implemented the presented strategy improvement algorithm in OCaml 3.10.2.
Our prototype uses Yices 1.0.27 \cite{yices,DBLP:conf/cav/DutertreM06} 
for computing models of SAT modulo real linear arithmetic formulas.
For solving the occurring linear programming problems we use QSOpt-Exact 2.5.6 \cite{DBLP:journals/orl/ApplegateCDE07,espinoza2006}.
QSOpt-Exact is an exact arithmetic version of QSOpt.
We made our experiments under Debian Linux (Lenny) on an Apple MacBook 
(2.16 GHz Intel Core 2 Duo, 2GB 667 MHz DDR2 SDRAM).%

%Our prototypic implementation together with a set of benchmarks can be downloaded under
%\begin{center}
%  https://?????????
%\end{center}
\end{comment}

\section{Conclusion}
\label{s:conc}

% PERFEKT
We presented an extension of the strategy improvement algorithm of 
\citet{DBLP:conf/esop/GawlitzaS07,DBLP:conf/csl/GawlitzaS07,gawlitza_sas_10}
which enables us to use a SAT modulo real linear arithmetic solver for
determining improvements of strategies w.r.t.\ current approximates.
Due to this extension,
we are able to compute abstract semantics 
of affine programs w.r.t.\ the template linear constraint domains
of \citet{DBLP:conf/vmcai/SankaranarayananSM05},
where we abstract sequences of if-then-else statements without loops en bloc.
This gives us additional precision.
Additionally,
We provided one of the few ``hard'' complexity results regarding precise abstract interpretation.
%
% We have implemented a prototype for this method and tested it on sample cases (see appendix \ref{a:exp:res}). 
%Operating directly on C source code is not yet possible, for want of a full front-end. 
%Unsurprisingly, the method is slowest on the artificial exponential examples described in Sec.~\ref{s:upper}.
%(See \url{http://tinyurl.com/37frok2}.)
%

% PERFEKT
It remains to practically evaluate the presented approach %in more detail
and to compare it systematically with other approaches.
Besides this, starting from the present work, 
there are several directions to explore.
One can for instance try to apply the same ideas for non-linear templates \citep{gawlitza_sas_10}, or to use linearization techniques~\citep{Mine_PhD}.
%Another is the problem of finding suitable templates --- while there exist obvious choices in some cases (intervals for getting rough invariants of control applications, difference bounds for scheduling applications, etc.), there is no generic method for obtaining good templates.

%In programs with pointers or references, the system of semantic equations on scalar values depends on a shape or alias analysis. 
%In some systems, such analysis is performed prior to the main analysis, but this seems sub-optimal and occasionally leads to blow-up when alias sets are too large. It thus seem desirable to perform both analyses at the same time, and thus to combine strategy iteration and more conventional Kleene iterations.
%For scalability purposes, it would also seem desirable to be able to ``slice'' fixpoint problems according to variable dependencies. 
\bibliographystyle{dmplainnat}
\bibliography{bib}
\end{document}